\documentclass[12pt]{article}
\usepackage{amsmath}
\usepackage{graphicx,psfrag,epsf}
\usepackage{enumerate}
\usepackage{natbib}
\usepackage{url} 
\usepackage{amsmath,amssymb,amsfonts,amsthm,mathtools} 
\usepackage{bm,bbm}
\usepackage{color}
\usepackage{subfigure} 
\usepackage{algorithm,algpseudocode}
\usepackage{scalefnt}

\newcommand{\blind}{0}

\newtheorem{rem}{{Remark}}
\newtheorem{prop}{{Proposition}}
\newtheorem{theorem}{{Theorem}}

\begin{document}

\def\spacingset#1{\renewcommand{\baselinestretch}%
{#1}\small\normalsize} \spacingset{1}


\if0\blind
{
  \title{\bf On a quantile autoregressive conditional duration model applied to high-frequency financial data}
  \author{  Helton Saulo\hspace{.2cm}\\
    \small Department of Statistics, University of Brasilia,  Brasilia, Brazil 70910-900\\
     \small Department of Mathematics and Statistics, McMaster University, Hamilton, Ontario, Canada L8S 4K1
    \\ \\
    Narayanaswamy Balakrishnan \\
   \small  Department of Mathematics and Statistics, McMaster University,  Hamilton, Ontario, Canada L8S 4K1
    \\
    and 
    \\
    Roberto Vila\\
    \small Department of Statistics, University of Brasilia,  Brasilia, Brazil 70910-900
    }
  \maketitle
} \fi

\if1\blind
{
  \bigskip
  \bigskip
  \bigskip
  \begin{center}
    {\Large\bf On a quantile autoregressive conditional duration model applied to high-frequency financial data}
\end{center}
  \medskip
} \fi

\bigskip
\vspace{-0.5cm}
\begin{abstract}
 Autoregressive conditional duration (ACD) models are primarily used to deal with data arising from times between two successive events. These models are usually specified in terms of a time-varying conditional mean or median duration. In this paper, we relax this assumption and consider a conditional quantile approach to facilitate the modeling of different percentiles. The proposed ACD quantile model is based on a skewed version of Birnbaum-Saunders distribution, which provides better fitting of the tails than the traditional Birnbaum-Saunders distribution, in addition to advancing the implementation of an expectation conditional maximization (ECM) algorithm. A Monte Carlo simulation study is performed to assess the behavior of the model as well as the parameter estimation method and to evaluate a form of residual. A real financial transaction data set is finally analyzed to illustrate the proposed approach.
\end{abstract}

\noindent%
\vfill

\newpage
\spacingset{1.45} 

\section{Introduction}\label{sec:01}

Transaction-level high-frequency financial data modeling has become increasingly important in the era of big data. In particular, duration data between successive events (trades, price changes, etc.) were initially modeled by the family of autoregressive conditional duration (ACD) models, proposed by \cite{er:98}. Since the inception 
of Engle and Russell's ACD model, several works have been published, supplying the literature with generalizations of the original model, associated inferential results, and diverse applications; see the literature review by \cite{pa:08} and its 
extension by \cite{bhogalvariyam:19}.

One of the most prominent ACD models in terms of fit and forecasting performance is the Birnbaum-Saunders distribution-based ACD (BS-ACD) model. This model was proposed by \cite{b:10} and as emphasized by \cite{bhogalvariyam:19}, it is constructed by taking into consideration the concept of conditional quantile estimation; the BS-ACD model is constructed in terms of a time-varying median duration; see also the recent discussion paper by \cite{balakundu:19}. Recently, \cite{saulolla:19} made a study and have shown that the BS-ACD model outperforms many other existing models in terms of model fitting and forecasting ability.

In this work, we propose a new ACD model by modeling the conditional quantile duration rather than the traditionally employed conditional mean (or median) duration. The quantile approach allows us to capture the influences of conditioning variables on the characteristics of the response distribution; see \cite{koenkexiao:06}. The proposed model is based on a skewed version of the BS (skew-BS) distribution, which has the main advantage over the classic BS distribution in its capability to fit data that are highly concentrated on the left-tail of the distribution, such is the case for transaction-level high-frequency financial data; the skew-BS distribution was proposed by \cite{vl:06} and inferential results for this model have been discussed by \cite{vslb:11}. We first introduce a reparameterization of the skew-BS model by inserting a quantile parameter, similarly to the work of \cite{sanchezetal:19} who introduced a quantile parameter in the BS distribution, and then develop the new ACD model, denoted by skew-QBS-ACD. 
We then demonstrate that the proposed skew-QBS-ACD model outperforms the BS-ACD model in terms of model fitting and forecasting ability, making it a practical and useful model.

The rest of this paper proceeds as follows. In Section \ref{sec:02}, we describe the usual skew-BS distribution and propose a reparameterization of this distribution in terms of a quantile parameter. In Section \ref{sec:03}, we introduce the skew-QBS-ACD model. In this section, we also describe the expectation conditional maximization (ECM) algorithm for the maximum likelihood (ML) estimation of the model parameters. In Section~\ref{sec:04}, we carry out a Monte Carlo simulation study to evaluate the performance of the estimators. In Section \ref{sec:05}, we apply the skew-QBS-ACD model to a real price duration data set, and finally in Section \ref{sec:06}, we provide some concluding remarks.

\section{Preliminaries}\label{sec:02}

In this section, we describe briefly the skew-BS distribution. We then introduce a quantile-based reparameterization of this distribution, which will be useful subsequently for developing the new ACD model. We derive some properties of the quantile-based skew-BS distribution; see Appendix A.

\subsection{The classical skew-BS distribution}
A random variable $Y$ follows a skew-BS distribution with shape parameter $\alpha>0$, scale parameter $\beta>0$ and skewness parameter $\lambda$, denoted by $Y\sim\textrm{skew-BS}(\alpha,\beta,\lambda)$, if its probability density function (PDF) is given by
\begin{equation}\label{eq:skewbs-pdf}
f_{Y}(y;\alpha,\beta,\lambda)
=2\phi\big[a(y;\alpha,\beta)\big]\Phi\big[\lambda a(y;\alpha,\beta)\big] a'(y;\alpha,\beta), \quad y>0,
\end{equation}
where
\begin{align}\label{at}
a(y;\alpha,\beta) 
= 
\frac{1}{\alpha} \left[\sqrt{\frac{y}{\beta}}-\sqrt{\frac{\beta}{y}}\,\right];
\quad 
a'(y;\alpha,\beta)
=
{1\over 2\alpha y}
\left[\sqrt{y\over \beta}+{\sqrt{\beta\over y}}\,\right];
\end{align}
and $\phi(\cdot)$ and $\Phi(\cdot)$ are the standard normal PDF and cumulative distribution function (CDF), respectively. 
\begin{rem}
Note that the \textrm{skew-BS} distribution \eqref{eq:skewbs-pdf} can be approached as a weighted distribution
\begin{align*}
f_{Y}(y;\alpha,\beta,\lambda)
={w(y) f_{T}(y;\alpha,\beta)\over \mathbb{E}[w(T)]}, \quad y>0,
\end{align*}
where $w(y)=\Phi\big[\lambda a(y;\alpha,\beta)\big]$ and $T$ follows a classic BS distribution, $T\sim\textrm{BS}(\alpha,\beta)$; see \cite{bs:69a} for details on the classic BS distribution.
\end{rem}

Figure \ref{figpdf:1} displays different shapes of the skew-BS PDF for different combinations of parameters. The plot show the clear effect the skewness parameter $\lambda$ has on the density function.

\begin{figure}[h!]
\vspace{-0.25cm}
\centering
{\includegraphics[height=6.2cm,width=6.2cm]{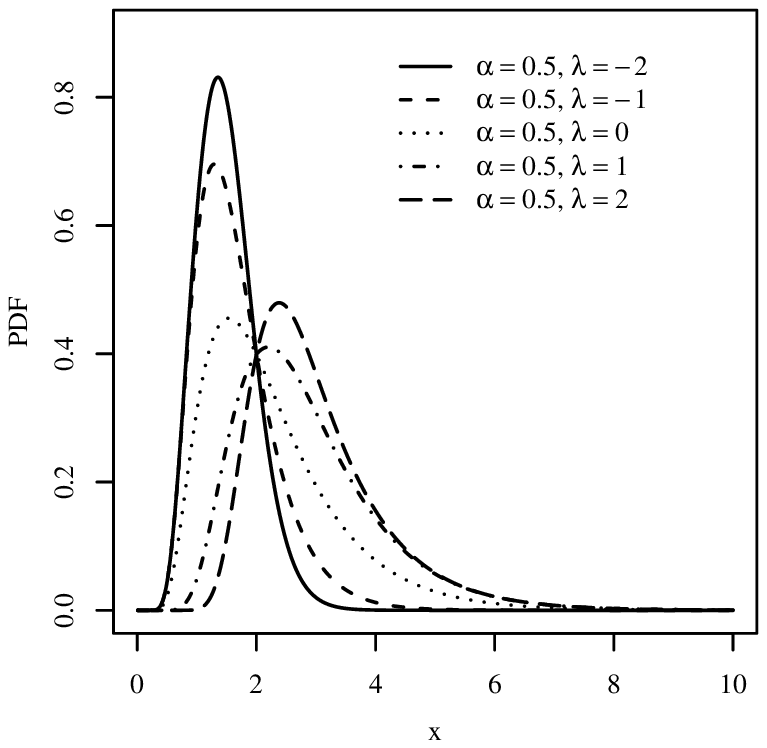}}
{\includegraphics[height=6.2cm,width=6.2cm]{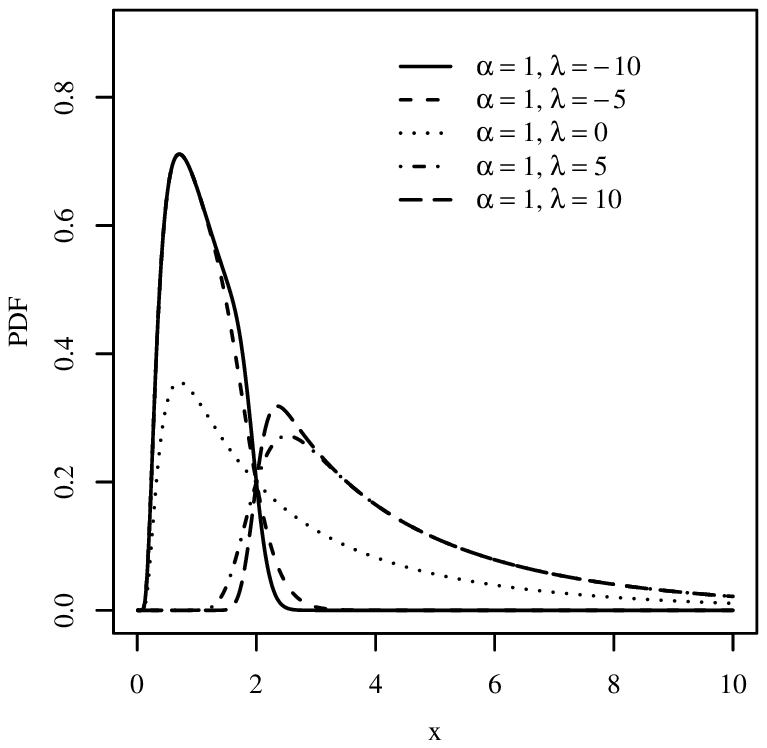}}\\
\vspace{-0.2cm}
\caption{Skew-QBS PDFs for some parameter values ($\beta=2.00$).}\label{figpdf:1}
\end{figure}

The random variable $Y$ can be stochastically represented as
\begin{equation}\label{eq:sr}
 Y=\frac{\beta}{4}\left[ \alpha X + 
 \sqrt{ (\alpha X)^2 +4   }    \right]^2,
 \quad 
 X=\delta U + Z \sqrt{1-\delta^2},
\end{equation}
where $\delta=\lambda/\sqrt{1+\lambda^2}$, $X$ is a random variable following a skew-normal (SN) distribution, $X\sim\textrm{SN}(\lambda)$, $Z$ is a standard normal random variable, $Z\sim\textrm{N}(0,1)$, and $U$ is a standard half-normal (HN) random variable,  $U\sim \textrm{HN}(0,1)$; see \cite{vslb:11}. The $100q$-th quantile of $Y\sim\textrm{skew-BS}(\alpha,\beta,\lambda)$ is given by
\begin{equation}\label{eq:qf}
 \Xi_{q}=Q_{Y}(q;\alpha,\beta,\lambda)
 =
 \frac{\beta}{4}\left\{ \alpha Q_X(q;\lambda)  + 
 \sqrt{[\alpha Q_X(q;\lambda)  ]^2 +4   }    \right\}^2.
\end{equation}
Interested readers may refer to the book by \cite{azzalinicap:14} for elaborate details on skew-normal distribution and related issues.

\subsection{The quantile-based skew-QBS distribution}

Consider a fixed number $q \in (0,1)$ and the following reparameterization $(\alpha, \beta,\lambda) \mapsto (\alpha, \Xi_{q},\lambda)$, where $\Xi_{q}$ is the $100q$-th quantile of $Y\sim\textrm{skew-BS}(\alpha,\beta,\lambda)$, as in \eqref{eq:qf}. Then, the PDF of $Y$ based on $\Xi_{q}$ is given by
\begin{align}\label{eq:skewbs-pdf-rep}
f_{Y}(y;\alpha,\Xi_{q},\lambda)
&=
2\phi\big[a(y;\alpha,\Xi_{q})\big]\Phi\big[\lambda a(y;\alpha,\Xi_{q}) \big] a'(y;\alpha,\Xi_{q}),
\quad y>0,
\end{align}
where
{
\color{black}
\begin{equation}\label{at-rep}
a(y;\alpha,\Xi_{q}) = \frac{1}{\alpha} \left[\sqrt{\frac{y \eta_{\alpha; \lambda}^2}{4\Xi_{q}}}-\sqrt{\frac{4\Xi_{q}}{y\eta_{\alpha; \lambda}^2}}\right];
\quad 
a'(y;\alpha,\Xi_{q})
=
{1\over 2\alpha y}
\left[\sqrt{\frac{y \eta_{\alpha; \lambda}^2}{4\Xi_{q}}}+\sqrt{\frac{4\Xi_{q}}{y\eta_{\alpha; \lambda}^2}}\right];
\end{equation}
with $\eta_{\alpha; \lambda}= \alpha Q_X(q;\lambda)  + 
\sqrt{[\alpha Q_X(q;\lambda) ]^2 +4   } $. 
}
Let us denote $Y\sim\textrm{skew-QBS}(\alpha,\Xi_{q},\lambda)$. Figure \ref{figpdf:2} illustrates some shapes provided by different values of $\Xi_{q}$ and $\lambda$ for the skew-QBS PDF. 

\begin{figure}[h!]
\vspace{-0.25cm}
\centering
{\includegraphics[height=6.2cm,width=6.2cm]{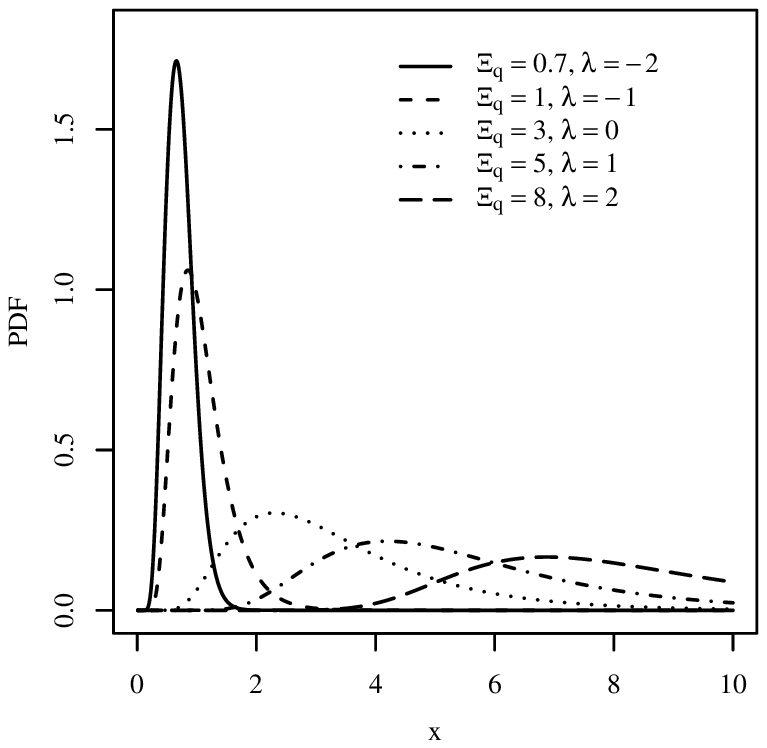}}
{\includegraphics[height=6.2cm,width=6.2cm]{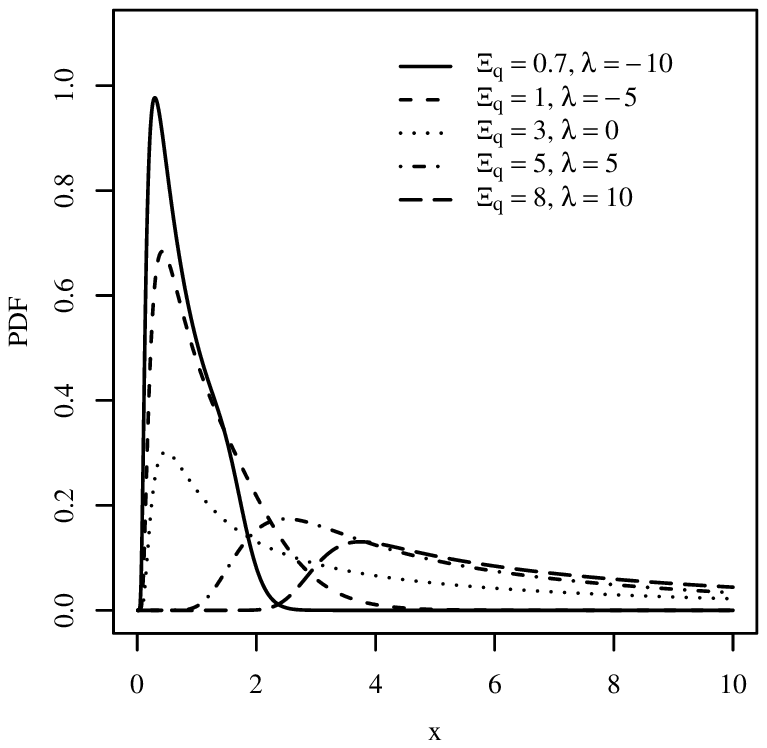}}\\
\vspace{-0.2cm}
\caption{Skew-QBS PDFs for some parameter values with $\alpha=0.50$ (left) and $\alpha=1.50$ (right) ($p=0.50$).}\label{figpdf:2}
\end{figure}

Note that, as in \cite{vslb:11}, the extended BS (EBS) distribution is obtained from the PDF of $Y|(U=u)$  
with $Y\sim\textrm{skew-QBS}(\alpha,\Xi_{q},\lambda)$ and $U\sim \textrm{HN}(0,1)$, which is given by
\begin{equation}\label{eq:skewbs-pdf-rep-ebs}
f_{Y|U}(y|u;\alpha_{\delta},\Xi_{q},\lambda_{h})
=\phi\big[\lambda_h + a(y;\alpha_{\delta},\Xi_{q})\big] a'(y;\alpha_{\delta},\Xi_{q}), \quad y>0,
\end{equation}
where $\alpha_{\delta}=\alpha\sqrt{1-\delta^2}$, $\lambda_{h}=-\delta h /\sqrt{1-\delta^2}$, and $a$ and $a'$ are as presented in \eqref{at} with $\alpha$ and $\lambda$ replaced by 
$\alpha_{\delta}$ and $\lambda_{h}$, respectively. We use the notation $Y\sim{\rm EBS}(\alpha_{\delta},\Xi_{q},\lambda_{h})$ in this case. Moreover, it is possible to obtain the conditional PDF of $U$, given $Y=t$, that is, the HN standard PDF: 
\begin{equation}\label{eq:skewbs-pdf-rep-uy}
f_{U|Y}(u|y;\alpha,\Xi_{q},\lambda,\delta)=
\frac{\phi\big[u;\mu=\delta a(y;\alpha,\Xi_{q}),\sigma^2=1-\delta^2 \big] }
{\Phi\big[ \lambda  a(y;\alpha,\Xi_{q}) \big]}, \quad u\geq 0,
\end{equation}
where $a(y;\alpha,\Xi_{q})$ is as given in \eqref{at}.

\section{The skew-QBS ACD model and ECM estimation}\label{sec:03}

\subsection{The skew-QBS ACD model}
Let $T_1,\ldots, T_n$ be a sequence of successive recorded times at which market events (trade, price change, etc.) occur. Then, define $Y_t=T_t-T_{t-1}$, for $i=1,\ldots,n$, that is, $Y_t$ is the time elapsed between two successive market events. We now propose a skew-QBS ACD model specified in terms of a time varying conditional quantile, $\Xi_{q,t} = F_{Y}^{-1}(q;\alpha,\Xi_{q,t},\lambda|\mathbb{F}_{t-1})$ say, where 
$F_{Y_t}^{-1}$ is the quantile function of the skew-QBS distribution and $\mathbb{F}_{t-1}$ is the set of all data available until time $T_{t-1}$. The conditional PDF of the skew-QBS ACD is given by 
\begin{equation}\label{eq:skewbs-pdf-acd-rep}
f_{Y|\Xi_{q}}(y_{t};\alpha,\Xi_{q,t},\lambda)
=
2\phi\big[a(y_t;\alpha,\Xi_{q,t},\lambda)\big]\Phi\big[\lambda a(y_t;\alpha,\Xi_{q,t},\lambda)\big] a'(y_t;\alpha,\Xi_{q,t},\lambda), \quad y_t>0,
\end{equation}
where 
\begin{align}\label{at-rep-1}
& a(y;\alpha,\Xi_{q,t},\lambda) = \frac{1}{\alpha} \left[\sqrt{\frac{y \eta_{\alpha; \lambda}^2}{4\Xi_{q,t}}}-\sqrt{\frac{4\Xi_{q,t}}{y\eta_{\alpha; \lambda}^2}}\right];
\\[0,2cm]
& a'(y;\alpha,\Xi_{q,t},\lambda) \label{at-rep-2}
=
{1\over 2\alpha y}
\left[\sqrt{\frac{y \eta_{\alpha; \lambda}^2}{4\Xi_{q,t}}}+\sqrt{\frac{4\Xi_{q,t}}{y\eta_{\alpha; \lambda}^2}}\right];
\end{align}
and
\begin{align}\label{definition-eta}
\eta_{\alpha; \lambda}= \alpha Q_X(q;\lambda)  + 
\sqrt{[\alpha Q_X(q;\lambda) ]^2 +4   }.
\end{align}
Here $Q_X(q;\lambda)$ denotes the $100q$-th quantile of $X\sim\textrm{SN}(\lambda)$.
Furthermore,
the dynamics of the conditional quantile are governed by
\begin{equation}\label{eq:skewbs-dyn-acd-rep}
 \log(\Xi_{q,t}) = \varpi + \sum_{j=1}^{r}\rho_{j} \log(\Xi_{q,t-j}) + \sum_{j=1}^{s}\frac{\sigma_{j}\,y_{t-j}}{\Xi_{q,t-j}}, \quad t=1,\ldots, n,
\end{equation}
implying the notation skew-QBS-ACD($r, s, q$). The ACD model in \eqref{eq:skewbs-pdf-rep} included as a special case the BS-ACD($r, s$) model proposed by \cite{b:10} when $q=0.5$ and $\lambda=0$. 
Note that the model in \eqref{eq:skewbs-pdf-acd-rep} can be represented as
\begin{equation}\label{eq:skewbs-pdf-acd-rep-formu}
Y_{t}=\Xi_{q,t}\,\varrho_{t}, \quad t = 1, \ldots, n,
\end{equation}
where $\varrho_{t}$ are independent and identically distributed (iid) random variables following the skew-QBS distribution, that is, $\varrho_{t} \stackrel{\text{iid}}{\sim}\textrm{skew-QBS}(\alpha,\Xi_{q}=1,\lambda)$ (see Proposition \ref{prop:proper}), and then the $Y_t$s are independent (ind) and non-identically distributed, i.e., $Y_t\stackrel{\text{ind}}{\sim} \textrm{skew-QBS}(\alpha,\Xi_{q,t},\lambda)$. 

In what follows we establish some results on moments, dispersion index function and autocorrelation function of the skew-QBS ACD model.
To this end, let us define the following:
\[
\boldsymbol{\Omega}=
\begin{pmatrix}
\rho_1 & \rho_2 & \ldots & \rho_{p-1} & \rho_p \\
1 & 0 & \ldots & 0 & 0 \\
0 & 1 & \ldots & 0 & 0 \\
0 & 0 & \ldots & 1 & 0
\end{pmatrix},
\quad 
\phi_k=\boldsymbol{\rho}^\top \boldsymbol{\Omega}^{k-p-1} \boldsymbol{\phi}, \
k>p=\max\{r,s\},
\]
where $\boldsymbol{\rho}=(\rho_1,\ldots,\rho_p)^\top$ and 
$\boldsymbol{\phi}=(\phi_1,\ldots,\phi_p)$  such that
\[
\phi_0=1, \quad \phi_1=\rho_1, \quad 
\phi_l=
\begin{cases}
\sum_{j=1}^l \rho_j \phi_{l-j}, & l=2,\ldots,p,
\\
\sum_{j=1}^p \rho_j \phi_{l-j}, & l>p.
\end{cases}
\]
Using the above notation, the next result shows that
the moments of the skew-QBS ACD model depend 
on the moments of $\varrho_{t}$.
\begin{prop}\label{prop-first}
Assume that $\mathbb{E}[{\rm e}^{m\theta_j \varrho_{t}} ]$ and
$\mu_m=\mathbb{E}[Y_t^m]$ exist for each $m>0$. For the  skew-QBS ACD model,
$\lambda_{\rm max}<1$ iff $\mu_m$ exists, where $\lambda_{\rm max}$ is the
absolute value of the maximum eigenvalue of the matrix $\boldsymbol{\Omega}$. 
Furthermore,
\begin{enumerate}
  \item \text{Moments of the skew-QBS ACD model:}
\[
\mathbb{E}[Y_t^m]
=
\mu_m
{\rm e}^{m\varpi(1-\sum_{j=1}^p\rho_j)^{-1}}
\,
\prod_{j=1}^\infty \mathbb{E}[{\rm e}^{m\theta_j \varrho_{t}}],
\]
where
$
\theta_1=\sigma_1 \ \text{and} \
\theta_l=
\begin{cases}
\sum_{j=1}^l \sigma_j \phi_{l-j}, & l=2,\ldots,p,
\\
\sum_{j=1}^p \sigma_j \phi_{l-j}, & l>p.
\end{cases}
$
  \item \text{Dispersion index function:}
\[
1+\delta_Y
=
(1+\delta^2) 
{
\prod_{j=1}^\infty \mathbb{E}[ {\rm e}^{2\theta_j  \varrho_t } ]
\over 
\big\{
\prod_{j=1}^\infty \mathbb{E}[{\rm e}^{\theta_j  \varrho_t }]
\big\}^2
}
\geqslant 
1+\delta^2,
\]
where $\delta=\sqrt{\rm Var(\varrho_{t})}/\mathbb{E}[\varrho_{t}]$ is the dispersion index of $\varrho_{t}$ and $\delta_Y$ denotes the degree of dispersion of the random variable $Y$ by the coefficient of variation.
\end{enumerate}
\end{prop}
\begin{proof}
The proof follows the same steps as those of Theorem 1 and Corollary 1 in 
\cite{bauwensetal:03}, by taking 
\[
\psi_t=\log(\Xi_{q,t}), \ \ \
\varpi=\omega,  \ \ \
\alpha_j=\sigma_j, \ \ \
g(\epsilon_t)=\epsilon_t =\varrho_{t}={y_t\over\Xi_{q,t}}, \ \ \
\beta_j=\rho_j \ \ \
\text{and} \ \  \ p=\max\{r,s\}
\] 
in Equation (5) of reference \cite{bauwensetal:03}.  
\end{proof}
\begin{prop}
Assume that $\mu=\mathbb{E}[\varrho_{t}]<\infty$, $\lambda_{\rm max}<1$,
$\mathbb{E}[{{\rm e}^{\delta \varrho_{t}}}]<\infty$ for any $\delta\in\mathbb{R}$, 
$\mathbb{E}[\varrho_{t-n} {\rm e}^{\theta_n \varrho_{t-n}}]<\infty$
for any $n\in\mathbb{N}_+$, and
$\mathbb{E}[{\rm e}^{(\phi_{n-j}\sigma_{j+h}+\theta^*_{hn})
(\varrho_{t-n-h})}]<\infty$ for $j$ and $h$ such that $n\geqslant 1$, where $\lambda_{\rm max}$ is as defined in Proposition \ref{prop-first}.
Then, for $n\geqslant 1$, the $n$-th order autocorrelation of 
$\{Y_t\}$ has the form
\begin{align*}
\rho_n
=
{
\mu \mathbb{E}[\varrho_{t} {\rm e}^{\theta_n \varrho_{t}}]
\prod_{j=1}^{n-1} 
\mathbb{E}[{{\rm e}^{\theta_j \varrho_{t}}}]
\prod_{j=p}^{\infty} 
\mathbb{E}[{{\rm e}^{\theta^*_{jn} \varrho_{t}}}] M_{n,p}
-
\mu^2
(\prod_{j=1}^\infty \mathbb{E}[{{\rm e}^{\theta_{j} \varrho_{t}}}])^2
\over
\mu_2 \prod_{j=1}^\infty \mathbb{E}[{{\rm e}^{2\theta_{j} \varrho_{t}}}]
-
\mu^2 (\prod_{j=1}^\infty \mathbb{E}[{{\rm e}^{\theta_{j} \varrho_{t}}}])^2
},
\end{align*}
where
\begin{align*}
M_{n,p}
=
\begin{cases}
\prod_{h=1}^{p-n}
\mathbb{E}[{\rm e}^{(\sum_{j=1}^n\phi_{n-j}\sigma_{j+h}+\theta^*_{hn})(\varrho_{t-n-h})}]
\prod_{h=1}^{n-1}
\mathbb{E}[{\rm e}^{(\sum_{j=1}^h\phi_{n-j}\sigma_{p-h+j}+\theta^*_{p-h,n})(\varrho_{t-n-h})}],
& n \leqslant p
\\
\prod_{h=1}^{p-1}
\mathbb{E}[{\rm e}^{(\sum_{j=1}^{p-h}\phi_{n-j}\sigma_{h+j}+\theta^*_{hn})(\varrho_{t-n-h})}],
& n> p
\end{cases},
\end{align*}
$
\theta^*_{jn}
=
\begin{cases}
\sum_{j=1}^h \sigma_j \phi^*_{h+1-j,n}, & h=1,\ldots,p
\\
\sum_{j=1}^p \sigma_j \phi^*_{h+1-j,n}, & h>p
\end{cases},
$ 
$\phi^*$ is as in (56) of \cite{bauwensetal:03},
and $\mu_2={\rm Var}(\varrho_{t})+\mu^2$.
\end{prop}
\begin{proof}
The proof follows from Theorem 2 of \cite{bauwensetal:03} by using the same notation as in the proof of Proposition \ref{prop-first}.
\end{proof}

\subsection{Estimation via the ECM algorithm}\label{sec:3-2}

\subsubsection{ML approach}
Let $\bm Y = (Y_{1}, \ldots, Y_{n})^\top$ be a random sample from $Y_{t}\sim \textrm{skew-QBS}(\alpha,\Xi_{q,t}, \lambda)$, for $t=1,\ldots,n$, with PDF as given in
\eqref{eq:skewbs-pdf-rep}, and ${\bm y} = (y_{1},\ldots, y_{n})^{\top}$ be the corresponding observations.
To obtain the ML estimates of the model parameters, $\bm{\theta} = (\alpha, \varpi, \rho_{1},\ldots,\rho_{r}, \sigma_{1}, \ldots, \sigma_{s}, \lambda)^{\top}$ say, we can maximize the log-likelihood function
\begin{equation}\label{eq:log:lik}
\ell(\bm{\theta})=
\sum_{t=1}^{n}
\log\big[
f_{Y|\Xi_{q}}(y_{t};\alpha,\Xi_{q,t},\lambda)
\big],
\end{equation}
where $f_{Y|\Xi_{q}}(y_{t};\alpha,\Xi_{q,t},\lambda)$ and $\Xi_{q,t}$  are as in \eqref{eq:skewbs-pdf-acd-rep} and \eqref{eq:skewbs-dyn-acd-rep}, respectively. As usual, it can be done by 
equating the score vector $\dot{\bm\ell}(\bm{\theta})$ to zero, and by using an iterative procedure for non-linear optimization, such as the Broyden-Fletcher-Goldfarb-Shanno (BFGS) quasi-Newton method. 

Let $f_{T|\Xi_{q}}(y_{t};\alpha,\Xi_{q,t},\lambda)
=
\phi\big[a(y;\alpha,\Xi_{q,t},\lambda)\big]a'(y;\alpha,\Xi_{q,t},\lambda)$ be the conditional PDF of the QBS distribution, where $a(y;\alpha,\Xi_{q,t},\lambda)$ and $a'(y;\alpha,\Xi_{q,t},\lambda)$ are as in \eqref{at-rep-1} and \eqref{at-rep-2}, respectively. Then, by using \eqref{eq:skewbs-pdf-acd-rep},  the conditional PDF of the skew-QBS ACD can be written as
$$
f_{Y|\Xi_{q}}(y_{t};\alpha,\Xi_{q,t},\lambda)
=
2
\Phi\big[\lambda a(y_t;\alpha,\Xi_{q,t},\lambda)\big]
f_{T|\Xi_{q}}(y_{t};\alpha,\Xi_{q,t},\lambda). 
$$
Using the above relation, note that
the elements of the score vector $\dot{\bm\ell}(\bm{\theta})$ are given by
\begin{align*}
{\partial \ell(\bm{\theta})\over \partial\gamma}
=
\sum_{t=1}^{n}
\dfrac{1}{f_{Y|\Xi_{q}}(y_{t};\alpha,\Xi_{q,t},\lambda)} \,
\dfrac{\partial f_{Y|\Xi_{q}}(y_{t};\alpha,\Xi_{q,t},\lambda)}{\partial \gamma},
\end{align*}
for each $\gamma\in\{\alpha, \varpi, \rho_{1},\ldots,\rho_{r}, \sigma_{1}, \ldots, \sigma_{s}, \lambda\}$, where
{	\scalefont{0.85}
\begin{align*}
{\partial f_{Y|\Xi_{q}}(y_{t};\alpha,\Xi_{q,t},\lambda)\over \partial\gamma}
=&
2\lambda
\phi\big[\lambda a(y_t;\alpha,\Xi_{q,t},\lambda)\big]
f_{T|\Xi_{q}}(y_{t};\alpha,\Xi_{q,t},\lambda) \,
{\partial a(y_t;\alpha,\Xi_{q,t},\lambda)\over \partial\gamma}
\\[0,2cm]
&+
2 \Phi\big[\lambda a(y_t;\alpha,\Xi_{q,t},\lambda)\big] \,
{\partial f_{T|\Xi_{q}}(y_{t};\alpha,\Xi_{q,t},\lambda)\over \partial\gamma},
\\[0,2cm]
{\partial f_{T|\Xi_{q}}(y_{t};\alpha,\Xi_{q,t},\lambda)\over \partial\gamma}
=&
\phi\big[a(y_t;\alpha,\Xi_{q,t},\lambda)\big]
\\[0,2cm]
& \times 
\left\{
{\partial a'(y_t;\alpha,\Xi_{q,t},\lambda)\over\partial\gamma}
-
a(y_t;\alpha,\Xi_{q,t},\lambda) a'(y_t;\alpha,\Xi_{q,t},\lambda) \, 
{\partial a(y_t;\alpha,\Xi_{q,t},\lambda)\over\partial\gamma}
\right\},
\end{align*}
}
for each $\gamma\in\{\alpha, \varpi, \rho_{1},\ldots,\rho_{r}, \sigma_{1}, \ldots, \sigma_{s},\lambda\}$. The above first order partial derivatives of functions $a(y_t;\alpha,\Xi_{q,t},\lambda)$ and $a'(y_t;\alpha,\Xi_{q,t},\lambda)$ defined in \eqref{at-rep-1} and \eqref{at-rep-2}, respectively, are written as
\begin{align*}
\textstyle
{\partial a(y_t;\alpha,\Xi_{q,t},\lambda)\over\partial\alpha}
&= \textstyle
-{a(y_t;\alpha,\Xi_{q,t},\lambda)\over \alpha}
+
2 a'(y_t;\alpha,\Xi_{q,t},\lambda) \eta_{\alpha; \lambda} \,
{\partial \eta_{\alpha; \lambda}\over \partial\alpha},
\\[0,2cm]
\textstyle
{\partial a(y_t;\alpha,\Xi_{q,t},\lambda)\over\partial\gamma^*}
&= \textstyle
-
a'(y_t;\alpha,\Xi_{q,t},\lambda) \,
{1\over \Xi_{q,t}^2} \,
{\partial \Xi_{q,t}\over \partial \gamma^*};
\quad 
\gamma^*\in\{\varpi, \rho_{1},\ldots,\rho_{r}, \sigma_{1}, \ldots, \sigma_{s}\} ,
\\[0,2cm]
\textstyle
{\partial a(y_t;\alpha,\Xi_{q,t},\lambda)\over\partial\lambda}
&= \textstyle
2 a'(y_t;\alpha,\Xi_{q,t},\lambda) \,
\eta_{\alpha; \lambda} \,
{\partial \eta_{\alpha; \lambda}\over \partial\lambda},
\end{align*}
and
\begin{align*}
\textstyle
{\partial a'(y_t;\alpha,\Xi_{q,t},\lambda)\over\partial\alpha}
&= \textstyle
-{a'(y_t;\alpha,\Xi_{q,t},\lambda)\over 2\alpha y_t}
+
2 a''(y_t;\alpha,\Xi_{q,t},\lambda)
\eta_{\alpha; \lambda}\,
{\partial \eta_{\alpha; \lambda}\over \partial\alpha} ,
\\[0,2cm]
\textstyle
{\partial a'(y_t;\alpha,\Xi_{q,t},\lambda)\over\partial\gamma^*}
&= \textstyle
-
a''(y_t;\alpha,\Xi_{q,t},\lambda) \,
{1\over \Xi_{q,t}^2} \,
{\partial \Xi_{q,t}\over \partial \gamma^*},
\quad 
\gamma^*\in\{\varpi, \rho_{1},\ldots,\rho_{r}, \sigma_{1}, \ldots, \sigma_{s}\} ,
\\[0,2cm]
\textstyle
{\partial a'(y_t;\alpha,\Xi_{q,t},\lambda)\over\partial\lambda}
&= \textstyle
2a''(y_t;\alpha,\Xi_{q,t},\lambda) \eta_{\alpha; \lambda}
\,
{\partial \eta_{\alpha; \lambda}\over \partial\lambda},
\end{align*}
where $\eta_{\alpha; \lambda}$ is as in \eqref{definition-eta} and
\begin{align}\label{sec-der-a}
a''(y_t;\alpha,\Xi_{q,t},\lambda) 
=
-
{1\over 4\alpha y^2_t}
\left[\sqrt{\frac{y_t \eta_{\alpha; \lambda}^2}{4\Xi_{q,t}}}+3\sqrt{\frac{4\Xi_{q,t}}{y_t\eta_{\alpha; \lambda}^2}}\right].
\end{align}
Furthermore, by using \eqref{eq:skewbs-dyn-acd-rep}, the above partial derivatives of quantile $\Xi_{q,t}$ are given by
\begin{align*}
\textstyle
{\partial \Xi_{q,t}\over\partial \varpi}
&= \textstyle
\Xi_{q,t}
\big[
1+\sum_{j=1}^{r}\rho_j\, {1\over \Xi_{q,t-j}}\, {\partial \Xi_{q,t-j}\over\partial \varpi}
-
\sum_{j=1}^{s} \sigma_j y_{t-j}\, {1\over \Xi^2_{q,t-j}} \, {\partial \Xi_{q,t-j}\over\partial \varpi}
\big],
\\[0,2cm] 
\textstyle
{\partial \Xi_{q,t}\over\partial \rho_{k}}
&= \textstyle
\Xi_{q,t}
\big[
\rho_k\log(\Xi_{q,t-k}) + 
\sum_{j=1}^{r}\rho_j\, {1\over \Xi_{q,t-j}}\, {\partial \Xi_{q,t-j}\over\partial \rho_k}
-
\sum_{j=1}^{s} \sigma_j y_{t-j}\, {1\over \Xi^2_{q,t-j}} \, {\partial \Xi_{q,t-j}\over\partial \rho_k}
\big], \ k=1,\ldots,r,
\\[0,2cm]
\textstyle
{\partial \Xi_{q,t}\over\partial \sigma_{l}}
&= \textstyle
\Xi_{q,t}
\big[
\sum_{j=1}^{r}\rho_j\, {1\over \Xi_{q,t-j}}\, {\partial \Xi_{q,t-j}\over\partial \sigma_l}
+
y_{t-l}\, {1\over \Xi_{q,t-l}}
-
\sum_{j=1}^{s} \sigma_j y_{t-j}\, {1\over \Xi^2_{q,t-j}} \, {\partial \Xi_{q,t-j}\over\partial \sigma_l}
\big], \quad l=1,\ldots,s.
\end{align*}


%
%
%
%

Let $I(\boldsymbol{\theta_0})$ denote the expected Fisher information matrix,  where $\boldsymbol{\theta_0}$ is the true value of the population parameter vector. Let us consider the following regularity conditions:
\begin{itemize}
\item 
The conditional PDF \eqref{eq:skewbs-pdf-acd-rep} of the skew-QBS ACD is an injective function in $\boldsymbol{\theta}$.
\item 
The log-likelihood function
${\ell}(\bm{\theta})$ is three times differentiable with respect to the unknown parameter vector $\bm{\theta}$ and the third partial derivatives are continuous
in $\bm{\theta}$.
\item 
Let $\bm Y = (Y_{1}, \ldots, Y_{n})^\top$ be a random sample from $Y_{t}\sim \textrm{skew-QBS}(\alpha,\Xi_{q,t}, \lambda)$, for $t=1,\ldots,n$.
For any $\boldsymbol{\theta_0}\in\Theta$, there exists $\delta> 0$ and functions $H_{uvw}({\bm y})$, for all ${\bm y}$ in $\Sigma$ (the support of the data),
such that for $\bm{\theta}= (\alpha, \varpi, \rho_{1},\ldots,\rho_{r}, \sigma_{1}, \ldots, \sigma_{s}, \lambda)^{\top}=(\gamma_1,\gamma_2,\ldots,\gamma_{r+s+3})^{\top}$ and $u,v,w,z = 1, 2,\ldots, r+s+3$, 
\begin{align*}
\left\vert {\partial^3{\ell}(\bm{\theta})\over \partial\gamma_u  \partial\gamma_v\partial\gamma_w}
\right\vert
\leqslant 
H_{uvw}({\bm y}),
\quad \text{for all} \, {\bm y}\in\Sigma 
\, \text{and} \, \vert \gamma_z-\gamma_{0,z} \vert<\delta,
\end{align*}
whenever $\mathbb{E}\big[H_{uvw}({\bm Y})\big]$ is finite.
\item 
For all $\boldsymbol{\theta}\in\Theta$ and $u = 1, 2,\ldots, r+s+3$;
$\mathbb{E}\big[{\partial{\ell}(\bm{\theta})\over \partial\gamma_u}\big]=0$.
\item 
The expected Fisher information matrix  $I(\boldsymbol{\theta})$
is finite, symmetric and positive definite. Furthermore,
\begin{align*}
n\big[I(\boldsymbol{\theta})\big]_{uv}
=
\mathbb{E}\bigg[{\partial{\ell}(\bm{\theta})\over \partial\gamma_u}\, {\partial{\ell}(\bm{\theta})\over \partial\gamma_v}\bigg]
=
-
\mathbb{E}\bigg[{\partial^2{\ell}(\bm{\theta})\over \partial\gamma_u \partial\gamma_v}\bigg],
\end{align*}
for each $u, v= 1, 2,\ldots, r+s+3$.
\end{itemize}
Following \cite{Davison08}, under the above assumptions, it follows that
\begin{align*}
\sqrt{n}\big[I(\boldsymbol{\theta}_0)\big]^{1/2}
\big(\widehat{\bm \theta}_n-{\bm \theta_0}\big)
\underset{ n\rightarrow\infty}{
\xrightarrow{\hspace*{0,7cm}\rm d \hspace*{0,7cm}}
}
{\rm N}\big({\bm 0}, I_{(r+s+3)\times (r+s+3)}\big),
\end{align*}
where $\widehat{\bm \theta}_n$ is the MLE  of
the parameter of interest ${\bm \theta}$,  ${\bm 0}$ is the $(r+s+3)\times 1$ zero vector and $I_{(r+s+3)\times (r+s+3)}$ is the ${(r+s+3)\times (r+s+3)}$ identity matrix. In practice, one may approximate the expected Fisher information matrix by its observed version, which is obtained from the Hessian matrix; see Appendix B. Then, the corresponding diagonal elements of the inverse of the observed Fisher information matrix can be used to approximate the standard errors (SEs) of the estimates.

%

\subsubsection{ECM algorithm}

We can implement the ECM algorithm presented in \cite{vslb:11}. Specifically, let 
${\bm y}=(y_{1},\ldots,y_{n})^{\top}$ and ${\bm u}=(u_{1},\ldots,u_{n})^{\top}$ be the observed and missing data, respectively, with $\bm Y$ and $\bm U$ being their corresponding random vectors. Thus, the complete data vector is written as ${\bm y}_{c}=({\bm y}^{\top}, {\bm u}^{\top})^{\top}$. From \eqref{eq:skewbs-pdf-rep-ebs}, we obtain
\begin{equation}
Y_{t}|(U_{t}=u_{t}) \sim   \textrm{EBS}(\alpha_{\delta},\Xi_{q},\lambda_{h_t}) \quad\text{with}\quad U_{t} \sim \textrm{HN}(0,1), \quad t=1,\ldots,n.\label{eq:ecm:stoch}
\end{equation}
Then, the complete data log-likelihood function for the skew-QBS-ACD model, associated
with ${\bm y}_{c}=({\bm y}^{\top}, {\bm u}^{\top})^{\top}$, is given by (without the constant)
\begin{equation}\label{eq:loglik:ecm}
\ell^c(\bm{\theta})=\sum_{t=1}^{n}\left\{ -\log(\alpha_{\delta})+\log\left(\frac{y_t \eta_{q}^2 + 4 \Xi_{q,t} }{2\eta_{q}\sqrt{\Xi_{q,t}}}\right)-\frac{1}{2(1-\delta^2)}
\big[a(y_t;\alpha,\Xi_{q,t})-\delta h_t\big]^{2}\right\}.
\end{equation}

The implementation of the ECM algorithm requires the computation of the expected value of the complete log-likelihood equation in \eqref{eq:loglik:ecm}, conditional on $ \bm U = \bm u$, that is,
\begin{eqnarray}\label{eq:loglik:ecm:exp}
Q({\bm \theta}|\widehat{\bm \theta}) &=& E\left[\ell^c(\bm{\theta})|\bm U = \bm u \right]
\\[0,2cm]
\nonumber
&=& 
\sum_{t=1}^{n}\left\{ -\log(\alpha_{\delta})+\log\left(\frac{y_t \eta_{q}^2 + 4 \Xi_{q,t} }{2\eta_{q}\sqrt{\Xi_{q,t}}}\right)- \frac{b^2(y_t;\Xi_{q,t})}{2\alpha_{\delta}^2} + \frac{\lambda}{\alpha_{\delta}} b(y_t;\Xi_{q,t})\widehat{u}_{t}-\frac{\lambda^2}{2}\widehat{u}_{t}^2 \right\}.
\end{eqnarray}
where $b(y_t;\Xi_{q,t}) =  \left[\sqrt{\frac{y_t \eta_q^2}{4\Xi_{q,t}}}-\sqrt{\frac{4\Xi_{q,t}}{y_t\eta_{q}^2}}\right]$, $\lambda=\delta/\sqrt{1-\delta^2}$, and 
$$
\widehat{u}_t=E[U_t|Y_t=y_t;\widehat{\bm \theta}]=
\widehat{\delta} a(y_t;\widehat{\alpha},{\Xi}_{q,t})+
\sqrt{1-\widehat{\delta}^2}\Upsilon_{\Phi}\left[  \frac{\widehat{\delta} a(y_t;\widehat{\alpha},{\Xi}_{q,t})}{\sqrt{1-\widehat{\delta}^2}}  \right]
$$

$$
\widehat{u}_t^2=E[U_t^2|Y_t=y_t;\widehat{\bm \theta}]=
\widehat{\delta}^2 a^2(y_t;\widehat{\alpha},{\Xi}_{q,t})+
(1-\widehat{\delta}^2)+\widehat{\delta}\sqrt{1-\widehat{\delta}^2}
\Upsilon_{\Phi}\left[  \frac{\widehat{\delta} a(y_t;\widehat{\alpha},{\Xi}_{q,t})}{\sqrt{1-\widehat{\delta}^2}}  \right]
a(y_t;\widehat{\alpha},{\Xi}_{q,t}),
$$
$$
 \log(\Xi_{q,t}) = \widehat{\varpi} + \sum_{j=1}^{r}\widehat{\rho}_{j} \log(\Xi_{q,t-j}) + \sum_{j=1}^{s}\frac{\widehat{\sigma}_{j}\,y_{t-j}}{\Xi_{q,t-j}}, 
$$
with $\Upsilon_{\Phi}(x)=\phi(x)/\Phi(x)$ for $x\in\mathbb{R}$; see \cite{vslb:11}.

The steps to obtain the ML estimates of the skew-QBS-ACD model via the ECM algorithm have been summarized in Algorithm \ref{alg:1}. Starting values are required to initiate the procedure, particularly, $\widehat{\alpha}^{(0)}, \widehat{\varpi}^{(0)}, \widehat{\rho}_{1}^{(0)},\ldots,\widehat{\rho}_{r}^{(0)}, \widehat{\sigma}_{1}^{(0)}, \ldots, \widehat{\sigma}_{s}^{(0)}$. These are obtained from \cite{saulolla:19}, and then $\widehat{\lambda}^{(0)}$ from \cite{vslb:11}. The iterations of Algorithm \ref{alg:1} are repeated until a stopping criterion has been met, such as $\ell(\bm{\theta}^{v+1})-\ell(\bm{\theta}^{v})\leq \epsilon$, where $\epsilon$ is a prescribed small tolerance level.

\begin{algorithm}[h!]
\small\caption{ECM approach for the skew-QBS-ACD($r, s, q$) model}\label{alg:1}
\begin{algorithmic}[1]
\State {\bf E-step.}
Compute $\widehat{u}_{t}^{(v)}$ and $\widehat{u}_{t}^{2(v)}$, given ${\bm \theta} = \widehat{\bm
\theta}^{(r)}$, for $t=1,\ldots,n$, $v=1,2,\ldots$;

\State {\bf CM-step 1.}
Fix  $\widehat{\bm{\Delta}}^{(v)} = (\widehat{\varpi}^{(v)}, \widehat{\rho}_{1}^{(v)},\ldots,\widehat{\rho}_{r}^{(v)}, \widehat{\sigma}_{1}^{(v)}, \ldots, \widehat{\sigma}_{s}^{(v)})^{\top}$ and update $\widehat{\alpha}^{(v)}$ and $\widehat{ \delta}^{(v)}$ as
$$\widehat{\alpha}^{2(v+1)}=\frac{1}{n}\sum_{t=1}^{n}b^2(y_t;\Xi_{q,t})+
\left[1- \frac{1}{n}\sum_{t=1}^{n}\widehat{u}_t^{2(v)} \right]
\left[ \frac{\sum_{t=1}^{n}\widehat{u}_{t}^{(v)} b(y_t;\Xi_{q,t}) }{\sum_{t=1}^{n}\widehat{u}_t^{2(v)}} \right]^2\quad \textrm{and}$$
$$\widehat{\delta}^{(v+1)}=\frac{1}{\widehat{ \alpha}^{(v)}}\left[ \frac{\sum_{t=1}^{n}\widehat{u}_{t}^{(v)} b(y_t;\Xi_{q,t}) }{\sum_{t=1}^{n}\widehat{u}_t^{2(v)}} \right],$$
with
$$
 \log(\Xi_{q,t}) = \widehat{\varpi}^{(v)} + \sum_{j=1}^{r}\widehat{\rho}_{j}^{(v)} \log(\Xi_{q,t-j}) + \sum_{j=1}^{s}\frac{\widehat{\sigma}_{j}^{(v)}\,y_{t-j}}{\Xi_{q,t-j}};
$$

\State {\bf CM-step 2.}
Fix $\widehat{\alpha}^{(v+1)}$ and $\widehat{\delta}^{(v+1)}$ and update 
$\widehat{\bm{\Delta}}^{(r)}$ by maximizing \eqref{eq:loglik:ecm:exp}.
\end{algorithmic}
\end{algorithm}

\subsection{Residual analysis}

To assess the goodness of fit and departures from the assumptions of the model, we shall use the generalized Cox-Snell (GCS) residual, given by
\begin{equation}\label{eq_sec:coxsnell}
 R^\textrm{GCS}_t = -\log\big[\widehat{S}(y_{t}|\Omega_{t-1})\big],
\end{equation}
where $\widehat{S}(\cdot)$ denotes the survival function fitted to the ACD data. This residual has a standard exponential, EXP(1), distribution when the model is correctly specified. Then, quantile-quantile (QQ) plots with simulated envelope can be used to assess the fit.

\section{Monte Carlo simulation}\label{sec:04}

We present the results of two Monte Carlo simulation studies for the skew-QBS-ACD($r, s, q$) model. Results are presented only for order of the lags set as $p=1$ and $q=1$, since a higher order for skew-QBS-ACD models did not improve the model fit considerably; see related results in \citet{b:10} and \citet{saulolla:19}. The first study considers the evaluation of the performance of the ML estimation procedure developed in Algorithm \ref{alg:1}, while the second study evaluates the empirical distribution of the residuals.

\subsection{ML estimation via the ECM algorithm}\label{sec:4.1}

The first simulation scenario considers the following setting: sample size $n \in \{500, 1000, 2000\}$, vector of true parameters $(\varpi,\rho,\sigma,\alpha,\lambda) = (0.20,0.70,0.10,0.50,-0.50)$ and $(\varpi,\rho,\sigma,\alpha,\lambda) = (0.20,0.70,0.10,1.00,-0.50)$, and $q=\{0.20,0.50,0.80\}$, with $200$ Monte Carlo replications for each sample size. The skew-QBS-ACD samples were generated using the transformation 
\begin{equation*}
 Y_t=\frac{\Xi_{q,t}}{\eta_{q}^2}\left[ \alpha X_t + \sqrt{ (\alpha X_t)^2 +4   }    \right]^2,
\end{equation*}
where $\eta_{q}$ is as in \eqref{at-rep}, $\Xi_{q,t}$ is as in \eqref{eq:skewbs-dyn-acd-rep}, and $X_t$ is a SN random variable generated by the function \texttt{rsn} of the \texttt{R} package \texttt{sn}; see \cite{sn:19}.

The ML estimation results are presented in Table~\ref{tab:MC:1} wherein the empirical mean, bias,  root mean squared error (RMSE), coefficients of skewness, and coefficients of (excess) kurtosis are all reported. A look at the results in Table~\ref{tab:MC:1} allows us to conclude that, in general, as the sample size increases, the bias and $\textrm{RMSE}$ of the estimates decrease, as expected. Moreover, $\widehat{\varpi}$, $\widehat{\rho}$, $\widehat{\sigma}$, $\widehat{\alpha}$ and $\widehat{\lambda}$ seem all to be consistent and marginally asymptotically distributed as normal.

\begin{table}[!ht]
\scriptsize
\centering
\renewcommand{\tabcolsep}{0.05cm}
 \caption{\small Summary statistics from simulated skew-QBS-ACD data.}\label{tab:MC:1}
 \begin{tabular}{l rrr c rrrrrrr}
\hline
& \multicolumn{3}{c}{$q=0.20$}&& \multicolumn{3}{c}{$q=0.50$} && \multicolumn{3}{c}{$q=0.80$}\\
\cline{2-4} \cline{6-8} \cline{10-12}
Statistic  & $n=500$       & $n=1000$    & $n=2000$  &\,\,& $n=500$       & $n=1000$ & $n=2000$ &\,\,& $n=500$       & $n=1000$ & $n=2000$\\
\cline{2-4} \cline{6-8} \cline{10-12}\\[-0.35cm]
&  \multicolumn{3}{c}{$\widehat \varpi$}  &\,\,&\multicolumn{3}{c}{$\widehat \varpi$}\\
\cline{2-4} \cline{6-8} \cline{10-12}
 True value               & 0.2000 & 0.2000 & 0.2000 && 0.2000 & 0.2000 & 0.2000  && 0.2000 & 0.2000 & 0.2000 \\
  Mean                    & 0.2761&  0.2494 & 0.2027   && 0.3304 & 0.2459 & 0.2205  && 0.3116 & 0.2504 & 0.2135 \\ 
  Bias                    & 0.0761 & 0.0494 & 0.0027   && 0.1304 & 0.0459 & 0.0205  && 0.1116 & 0.0504 & 0.0135 \\ 
  Root Mean Squar. Error  & 0.3721 & 0.2820 & 0.0658   && 0.3894 & 0.2087 & 0.0920  && 0.3496 & 0.2267 & 0.1005 \\ 
  Coefficient of skewness & 4.7959 & 6.5433 & 0.4690   && 2.8714 & 4.9210 & 0.7686  && 2.1133 & 3.2815&  0.8670 \\ 
  Coefficient of kurtosis & 27.5630& 49.5154 & 2.6336  && 9.3422 & 34.9339 & 0.8334 && 7.6321 &18.5650 & 3.8926  \\ 
\cline{2-4} \cline{6-8} \cline{10-12}\\[-0.35cm]
&  \multicolumn{3}{c}{$\widehat \rho$}  &\,\,&\multicolumn{3}{c}{$\widehat \rho$} &\,\,&\multicolumn{3}{c}{$\widehat \rho$}\\
\cline{2-4} \cline{6-8} \cline{10-12}
 True value               & 0.7000 & 0.7000  & 0.7000    && 0.7000 & 0.7000  & 0.7000  && 0.7000 & 0.7000  & 0.7000 \\ 
  Mean                    & 0.6351 & 0.6590&  0.6991       && 0.5681 & 0.6543  & 0.6798  && 0.5674 & 0.6420 & 0.6855 \\ 
  Bias                    &-0.0649& -0.0410& -0.0009       &&-0.1319& -0.0457 & -0.0202  && -0.1326& -0.0580 &-0.0145 \\ 
  Root Mean Squar. Error  & 0.2912&  0.2211&  0.0581       && 0.3753 & 0.2010  & 0.0961  && 0.4077 & 0.2609&  0.1200 \\ 
  Coefficient of skewness &-4.4652& -6.1797& -0.4745       &&-2.4850& -4.2562 & -0.6326  && -1.8745 &-2.9022 &-0.8088 \\ 
  Coefficient of kurtosis & 25.1453& 45.8789 & 2.8597      && 6.7188 & 27.7012 & 0.7042  && 6.3473& 15.9944 & 4.0282 \\
\cline{2-4} \cline{6-8} \cline{10-12}\\[-0.35cm]
&  \multicolumn{3}{c}{$\widehat \sigma$}  &\,\,&\multicolumn{3}{c}{$\widehat \sigma$}  &\,\,&\multicolumn{3}{c}{$\widehat \sigma$}\\
\cline{2-4} \cline{6-8} \cline{10-12}
 True value               & 0.1000 & 0.1000 & 0.1000   && 0.1000 & 0.1000 & 0.1000  && 0.1000 & 0.1000 & 0.1000 \\ 
  Mean                    & 0.1021 & 0.1003&  0.0995     && 0.1066 & 0.1016 & 0.1008  && 0.1054 & 0.0993 & 0.0988 \\ 
  Bias                    & 0.0021 & 0.0003& -0.0005     && 0.0066 & 0.0016 & 0.0008  && 0.0054& -0.0007& -0.0012 \\ 
  Root Mean Squar. Error  & 0.0291 & 0.0204&  0.0122     && 0.0365 & 0.0248 & 0.0170  && 0.0526 & 0.0378&  0.0253 \\ 
  Coefficient of skewness &-0.8058 &-0.9762&  0.0715     &&-0.1877 & -0.3339 & 0.0317 && -0.1767 & 0.3540&  0.2581 \\ 
  Coefficient of kurtosis & 5.1543 & 8.1370&  3.0733     &&0.3655 & 1.5183 & -0.1561  && 3.6280 & 3.6456 & 2.7845 \\ 
\cline{2-4} \cline{6-8} \cline{10-12}\\[-0.35cm]
&  \multicolumn{3}{c}{$\widehat \alpha$}  &\,\,&\multicolumn{3}{c}{$\widehat \alpha$} &\,\,&\multicolumn{3}{c}{$\widehat \alpha$}\\
\cline{2-4} \cline{6-8} \cline{10-12}
 True value               & 0.5000 & 0.5000 & 0.5000   && 0.5000 & 0.5000 & 0.5000 && 1.0000 & 1.0000 & 1.0000 \\  
  Mean                    & 0.5729 & 0.5528 & 0.5466     && 0.5432 & 0.5172 & 0.5047 && 1.1255 & 1.0902 & 1.0469 \\ 
  Bias                    & 0.0729 & 0.0528 & 0.0466     && 0.0432 & 0.0172 & 0.0047 &&0.1255 & 0.0902 & 0.0469 \\ 
  Root Mean Squar. Error  & 0.0923 & 0.0675 & 0.0581     && 0.0773 & 0.0464 & 0.0304 && 0.2445 & 0.1736&  0.1174 \\ 
  Coefficient of skewness & 0.3083 & 0.3835 & 0.2141     && 0.6612 & 0.6130 & 0.5629 && 1.0625 & 1.0157&  0.9184 \\ 
  Coefficient of kurtosis & 2.4722 & 2.3807 & 2.3406     && 0.1477 & -0.1976& -0.2117 && 3.2276 & 3.6309&  2.8638 \\ 
\cline{2-4} \cline{6-8} \cline{10-12}\\[-0.35cm]
&  \multicolumn{3}{c}{$\widehat \lambda$}  &\,\,&\multicolumn{3}{c}{$\widehat \lambda$} &\,\,&\multicolumn{3}{c}{$\widehat \lambda$}\\
\cline{2-4} \cline{6-8} \cline{10-12}
  True value              & -0.5000 & -0.5000 & -0.5000 && -0.5000 & -0.5000 & -0.5000 &&-0.5000 & -0.5000 & -0.5000 \\      
  Mean                    & -0.9964 & -0.8655 &-0.8285     && -0.6447 & -0.5228 & -0.4382 && -0.6915 &-0.7408 &-0.6222 \\ 
  Bias                    & -0.4964& -0.3655 &-0.3285     && -0.1447 & -0.0228 & 0.0618  && -0.1915 &-0.2408& -0.1222 \\ 
  Root Mean Squar. Error  & 0.6272 & 0.4570 & 0.4013      && 0.6439  & 0.4280  & 0.3401  && 0.8838&  0.5625 & 0.3772 \\ 
  Coefficient of skewness & 0.7651& -0.1717 &-0.1820      && 0.4968  & 0.4101  & 0.6433  &&  1.1817 &-0.0401& -0.6527 \\ 
  Coefficient of kurtosis & 8.2048&  2.4254 & 2.2408      &&0.2460  &-0.5180  &-0.0385   && 8.2765 & 5.4238&  2.5382 \\ 
   \hline
\end{tabular}
\end{table}

\subsection{Empirical distribution of the residuals}\label{sec:4.2}

In this subsection, we present a Monte Carlo simulation study to evaluate the performance of the GCS residuals. The sample generation as well as the simulation scenario are the same ($\alpha=0.50$) in Subsection \ref{sec:4.1}.

Table \ref{tab:res:1} presents the empirical mean, standard deviation, coefficient of skewness and coefficient of kurtosis, whose values are expected to be 1, 1, 2 and 9, respectively, for $R^\textrm{GCS}$. Notice that the reference distribution for $R^\textrm{GCS}$ is standard exponential distribution. From Table \ref{tab:res:1}, we note that, in general, the considered residuals conform well with the reference distribution.

\begin{table}[!ht]
\footnotesize
\centering
\renewcommand{\tabcolsep}{0.05cm}
 \caption{\small Summary statistics of the $R^\textrm{GCS}$ residuals.}\label{tab:res:1}
 \begin{tabular}{l rrr c rrrrrrr}
\hline
Statistic & \multicolumn{3}{c}{$q=0.20$}&& \multicolumn{3}{c}{$q=0.50$} && \multicolumn{3}{c}{$q=0.80$}\\
\cline{2-4} \cline{6-8} \cline{10-12}
  & $n=500$               & $n=1000$    & $n=2000$  &\,\,& $n=500$       & $n=1000$ & $n=2000$ &\,\,& $n=500$       & $n=1000$ & $n=2000$\\
  Mean                    & 0.9963 & 0.9981 & 0.9992 && 0.9968 & 0.9985 & 0.9993  && 0.9962 & 0.9982 & 0.9993\\ 
  Standard deviation      & 1.0118 & 1.0091 & 1.0075 && 1.0046 & 1.0051 & 1.0036  && 1.0029 & 1.0001 & 1.0001\\ 
  Coefficient of skewness & 2.0449 & 2.0210 & 2.0344 && 1.9779 & 2.0052 & 2.0101  && 1.9915 & 1.9639 & 1.9904\\ 
  Coefficient of kurtosis & 9.1639 & 8.9867 & 9.2581 && 8.6297 & 8.8636 & 9.0717  && 8.7890 & 8.5955 & 8.9215\\ 
  \hline
\end{tabular}
\end{table}


\section{Application to price duration data}\label{sec:05}

\subsection{Exploratory data analysis}

In this section, a real high frequency financial data set corresponding to price durations of BASF-SE stock on 19th April 2016 is used to illustrate the proposed methodology; see \cite{saulolla:19}. The data were adjusted to remove intraday seasonality, namely, activity is low at lunchtime and high at the beginning and closing of the trading day. Table~\ref{tab:descp} provides descriptive statistics for both plain and adjusted price durations. From this table, we note that both series show positive skewness and a high degree of kurtosis.

\begin{table}[!ht]
\centering
\caption{Summary statistics for the BASF-SE data.}\label{tab:descp}
\begin{tabular}{lcccccccc}
\hline
BASF-SE data                    &    Plain     &  &   Adjusted   \\
\hline
$n$                                &   2194       &  &  2194   \\
Minimum                            &  1           &  &  0.061    \\
10th percentile                    &  1            &  &  0.127         \\
Mean                               &   12.292     &  &  1.067       \\
50th percentile (median)           &  7           &  &  0.682       \\
90th percentile                    &  28           &  &  2.438     \\
Maximum                            &  266         &  &  9.776        \\
Standard deviation                 & 16.156       &  &  1.167        \\
Coefficient of variation           & 131.43\%     &  &  109.35\%        \\
Coefficient of skewness            & 4.695        &  &     2.521     \\
Coefficient of excess kurtosis     & 39.516       &  &  5.902     \\
\hline
\end{tabular}
\end{table}

\subsection{Estimation and model validation}

Table \ref{tab:resulall} presents the ML estimates, computed by the ECM algorithm, and SEs for the skew-QBS-ACD model parameters with $q = \{0.13,0.50\}$. On the one hand, the value of $q = 0.13$ was chosen through the profile log-likelihood, that is, for a grid of values of $q$, $0.01, 0.02,\ldots,0.99$, we estimated the model parameters and computed the corresponding log-likelihood functions. Then, the optimal value of $q$ was the one which maximized the log-likelihood function. On the other hand, the value of $q = 0.50$ corresponds to the skew-QBS-ACD in terms of the conditional median duration, such as the BS-ACD model. Table \ref{tab:resulall} also reports the log-likelihood value, Akaike (AIC) and Bayesian information (BIC) criteria, and $p$-values of the Ljung-Box (LB) statistic, $Q(l)$ say, for up to $l$-th order serial correlation. The LB statistics evaluate the autocorrelation of the residuals. For comparative purpose, the results of the exponential ACD (EXP-ACD), generalized gamma ACD (GG-ACD) and BS-ACD models, are provided as well. The results of Table \ref{tab:resulall} reveal that 
the skew-QBS-ACD models with $q = \{0.13,0.50\}$ provide better adjustments compared to other models based on the values of log-likelihood, AIC and BIC. However, these values are not substantially different for the two values of $q$. Finally, the LB $p$-values provide no evidence of serial correlation in the residuals for all models considered in this application.

Figure \ref{fig:estimates} plots the estimated parameters of the skew-QBS-ACD model across $q\in\{0.01,\ldots, 0.99\}$. Results suggest that this BASF-SE series displays asymmetric dynamics: the estimates of ${\varpi}$ and ${\sigma}$ tend to be increasing across $q$, with a change of sign in the first case. On the other hand, the estimates of $\rho$, $\alpha$ and $\lambda$ appear to vary with the same pattern, but without an explicit trend.

 \begin{table}[!ht]
 \centering
\begin{small}
 \caption{{ML estimates (with SE in parentheses) and model selection measures for fit to the BASF-SE data.}}
 \begin{tabular}{l rrrrrrrrrrrrrrrrrrrrr}
\hline   
          & EXP-ACD& GG-ACD &  BS-ACD & skew-QBS-ACD         &  skew-QBS-ACD      \\[-0.15cm]
          &        &        &         &   ($q=0.50$)        &    ($q=0.13$)       \\ 
\hline
$\varpi$  &   0.0658  &  0.0724    &$-$0.2756    &$-$0.2362 &  -0.7723     \\[-0.1cm]
          & (0.0180)  & (0.0183)   & (0.0975)    & (0.0822) &   (0.3034) \\
$\rho$    &  0.6807   &  0.7204    &    0.5800   &  0.5948  &   0.5962 \\[-0.1cm]
          & (0.1203)  & (0.0936)   & 0.1759      & (0.1712) &  (0.1685)\\
$\sigma$  & 0.0937   &  0.1041     &    0.0403   &  0.0450  & 0.0119    \\[-0.1cm]
          & (0.0206) & (0.0211)    &  (0.0107)   & (0.0118) & (0.0031) \\
          &          &             &             &                             \\[-0.2cm]
$\alpha$  &          & 15.0013     & 1.1975      &  1.3060   &  1.3066       \\[-0.1cm]
          &          &(6.5247)     & (0.0180)    & (0.0633)  &  (0.0638)      \\
$\xi$     &          & 0.2467      &             &           &           \\[-0.1cm]
          &          &  (0.0547)   &             &           &               \\[0.2cm]
$\lambda$ &          &             &             &  0.5197   &  0.5213        \\[-0.1cm]
          &          &             &             & (0.1578)  &  (0.1586)              \\[0.2cm]         
Log-lik.  & -2315.916& -2239.192   & -2226.839   &  -2216.559  & -2216.516   \\
AIC       & 4637.832 & 4488.384    & 4461.659    &  4443.118  &  4443.033  \\
BIC       & 4654.912 & 4516.851    & 4484.433    &  4471.585  &  4471.500  \\
$Q(4)$    &  0.9451  & 0.9451      & 0.8152      &  0.8152    &  0.8141      \\
$Q(16)$   &  0.3468  & 0.3468      & 0.2566      &  0.2566    &  0.2565      \\
\hline
\end{tabular}
\label{tab:resulall}
\end{small}
\end{table}

\begin{figure}[!ht]
\centering
\subfigure[$\widehat{\varpi}$]{\includegraphics[height=5cm,width=5cm]{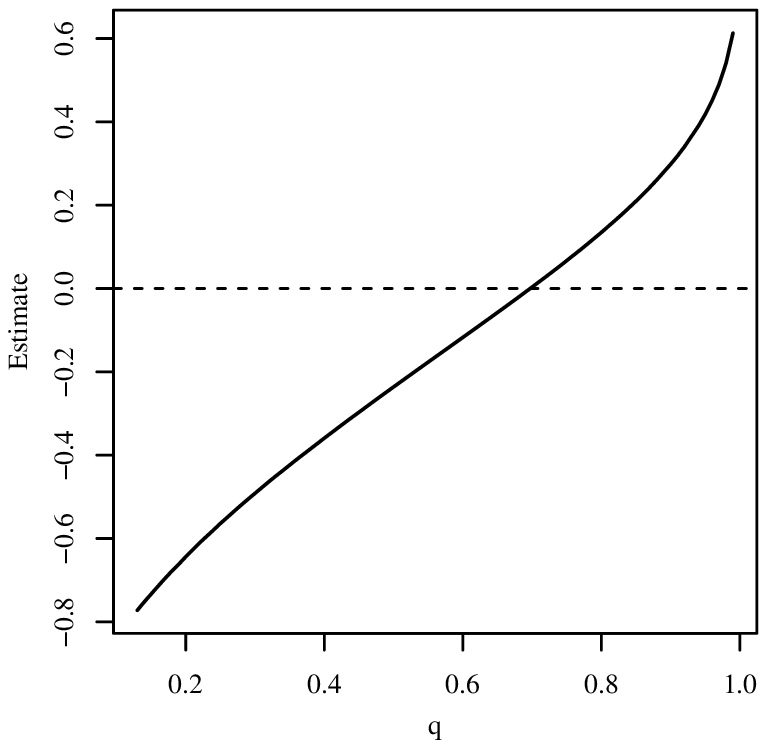}}
\subfigure[$\widehat{\rho}$]{\includegraphics[height=5cm,width=5cm]{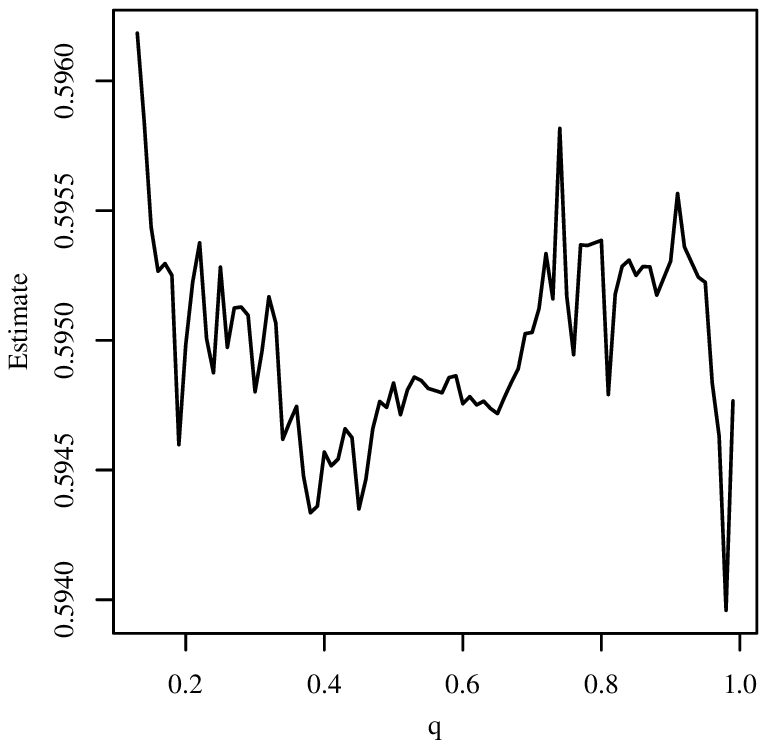}}
\subfigure[$\widehat{\sigma}$]{\includegraphics[height=5cm,width=5cm]{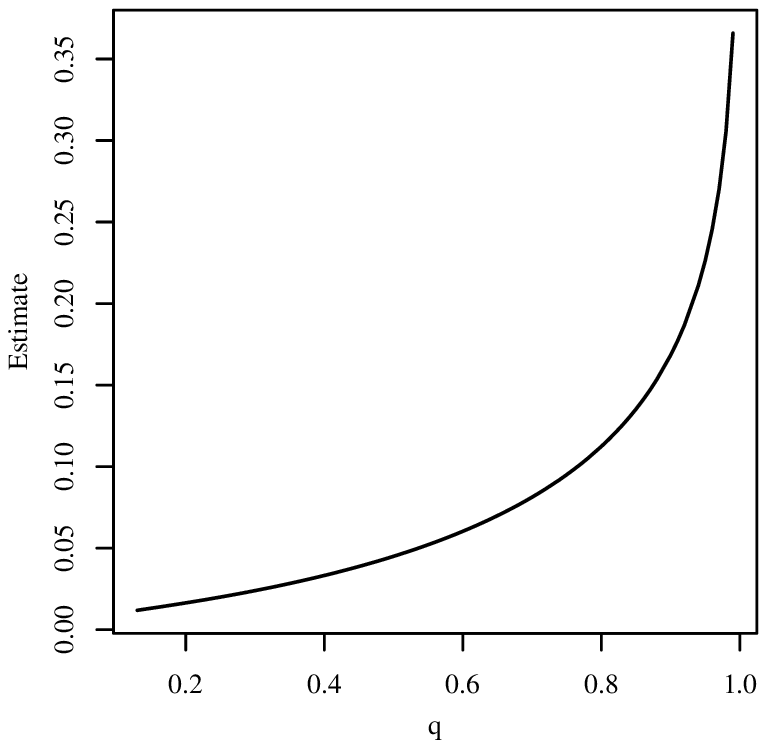}}\\
\subfigure[$\widehat{\alpha}$]{\includegraphics[height=5cm,width=5cm]{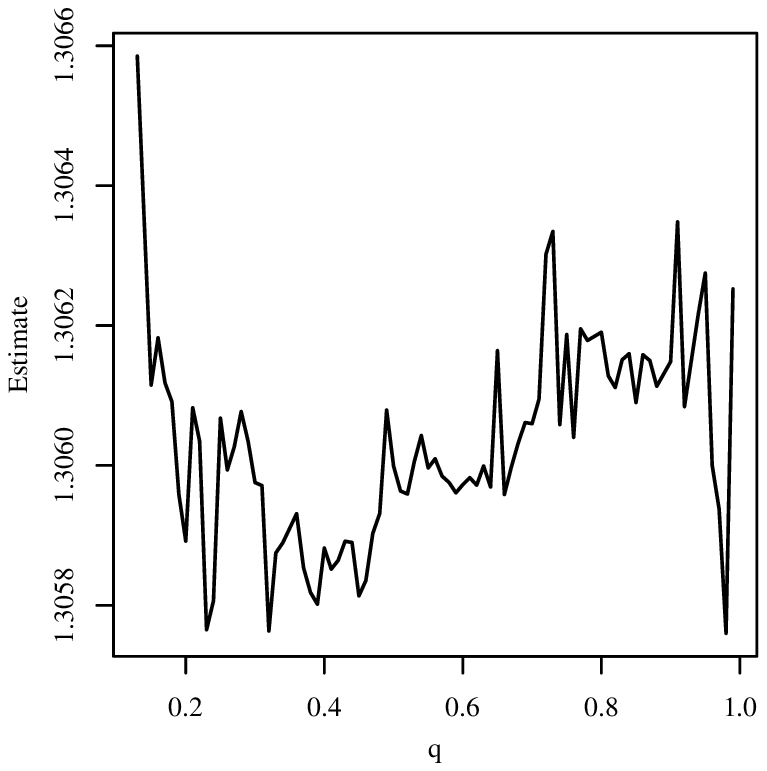}}
\subfigure[$\widehat{\lambda}$]{\includegraphics[height=5cm,width=5cm]{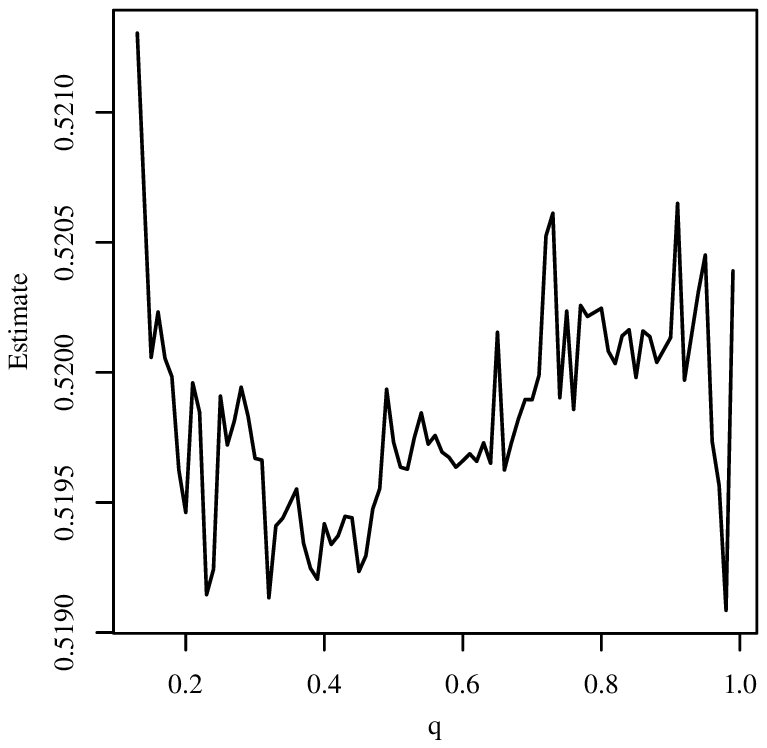}}
 \caption{\small {Estimated parameters in the skew-QBS-ACD model across $q$ for the BASF-SE data.}}
\label{fig:estimates}
\end{figure}

Figure~\ref{fig:qqplots} shows the QQ plots with simulated envelope of the generalized Cox-Snell residuals for the models considered in Table \ref{tab:resulall}. We see clearly that the skew-QBS-ACD models provide better fit than the EXP-ACD, GG-ACD  and BS-ACD models.

\begin{figure}[!ht]
\centering
\subfigure[EXP-ACD]{\includegraphics[height=5cm,width=5cm]{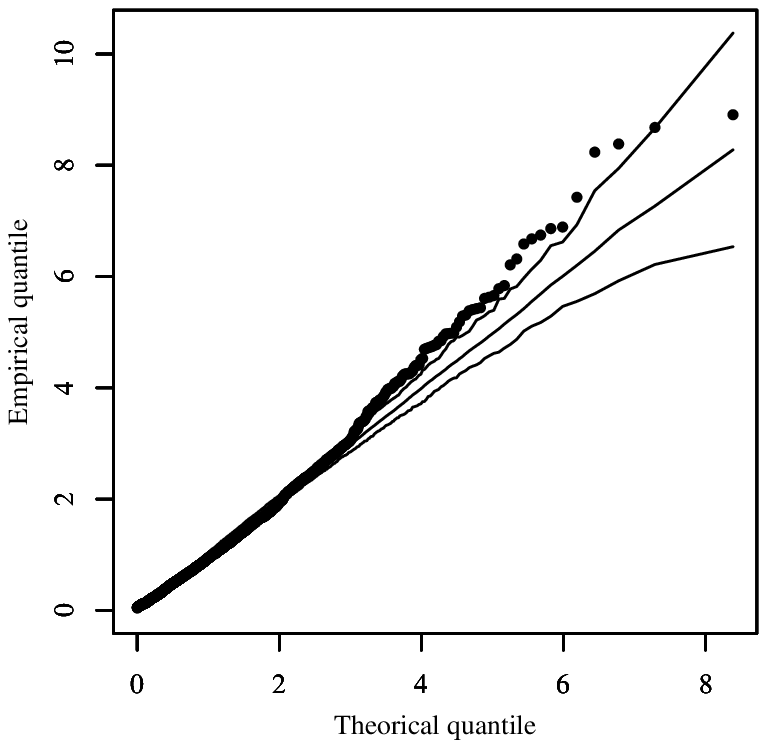}}
\subfigure[GG-ACD]{\includegraphics[height=5cm,width=5cm]{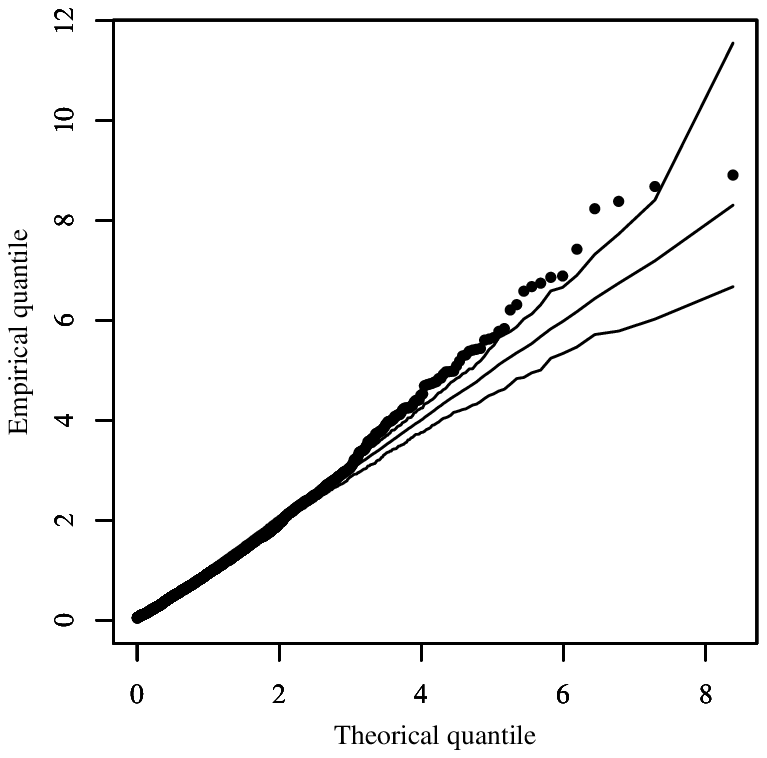}}
\subfigure[BS-ACD]{\includegraphics[height=5cm,width=5cm]{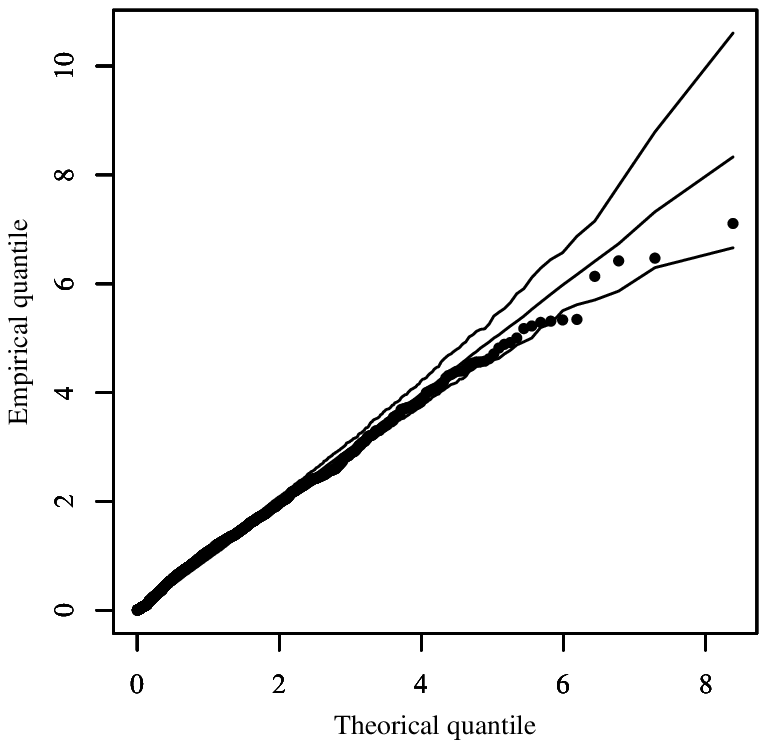}}\\
\subfigure[skew-QBS-ACD($q=0.50$)]{\includegraphics[height=5cm,width=5cm]{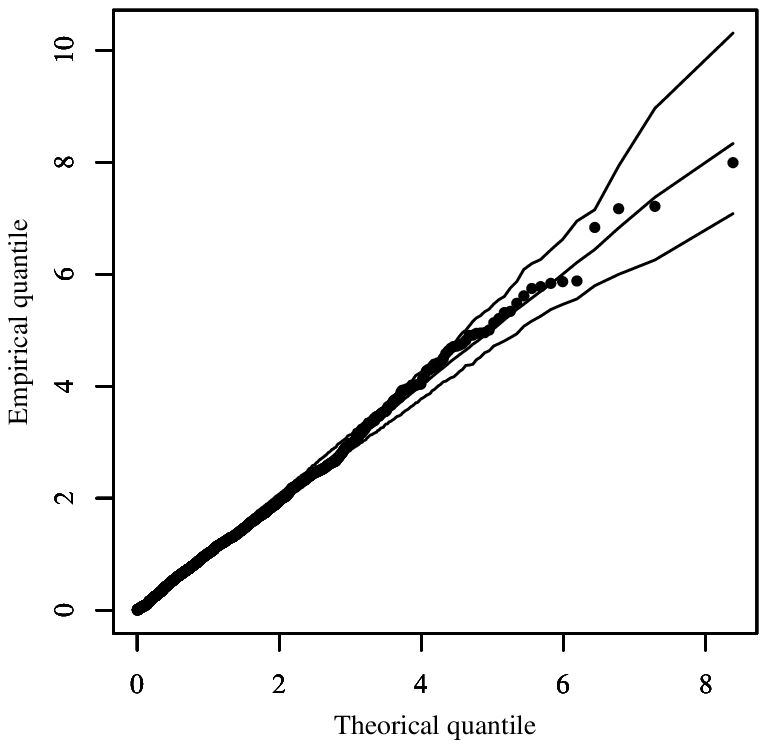}}
\subfigure[skew-QBS-ACD($q=0.13$)]{\includegraphics[height=5cm,width=5cm]{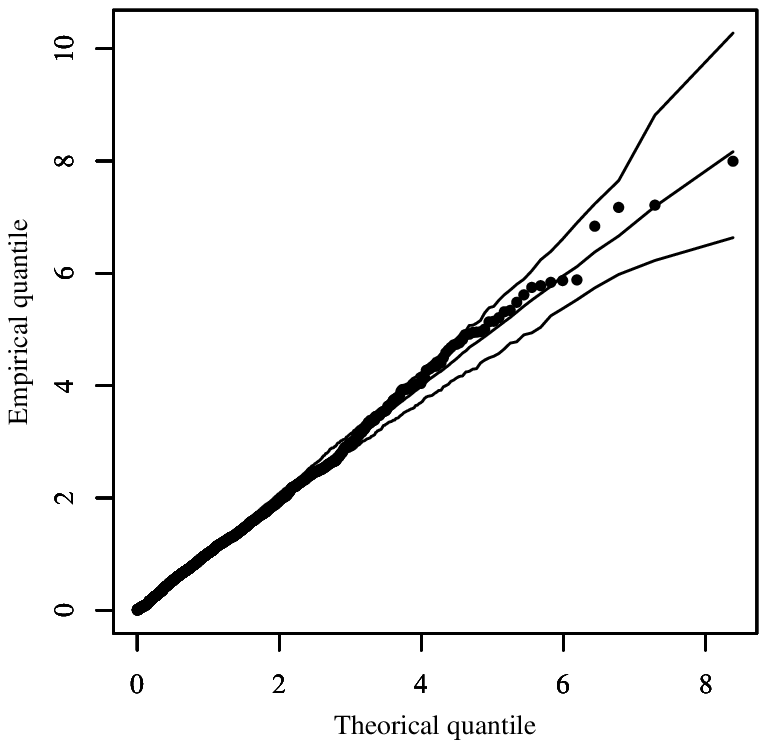}}
 \caption{\small {QQ plot and its envelope for the generalized Cox-Snell residual in the indicated model with the BASF-SE data.}}
\label{fig:qqplots}
\end{figure}

\subsection{Forecasting performance}

Figure \ref{figfore:1} plots the in-sample and out-of-sample quantile forecasts from the adjusted skew-QBS-ACD and BS-ACD models. The construction of the plots is as follows. For each value of $q\in\{0.01,\ldots, 0.99\}$, the skew-QBS-ACD model is estimated and the respective estimates are used to obtain the forecast of the quantile of interest. Since the BS-ACD model is independent of the value of $q$, it is estimated only once and its forecasts are obtained from this single estimation. For in-sample forecasting, all data are used. In the case of out-of-sample forecasting, 2/3 of the data are used for estimation and the other 1/3 for forecasting.

From Figures \ref{figfore:1}(a,b), we observe that the forecasts of the skew-BS-ACD model are much closer to the actual values than those of the BS-ACD model.  In fact, the forecasts of this latter model are quite far from the actual values. The superiority of the skew-QBS-ACD model may be due to different estimates obtained for the parameters; see, e.g., Figure \ref{fig:estimates}. For different quantiles, we obtain different estimates, which makes the forecasts closer to the actual values. The results of the estimated mean square error (MSE) between the actual data and in-sample forecasted values of the adjusted skew-QBS-ACD and BS-ACD models are 0.041 and 1.277, respectively. Moreover, the MSE between the actual data and out-of-sample forecasted values of the adjusted skew-QBS-ACD and BS-ACD models are 0.029 and 1.199, respectively. Thus, the MSE values do confirm the results found in Figures \ref{figfore:1}(a,b).

\begin{figure}[!ht]
\vspace{-0.25cm}
\centering
\subfigure[In-sample]{\includegraphics[height=6.2cm,width=6.2cm]{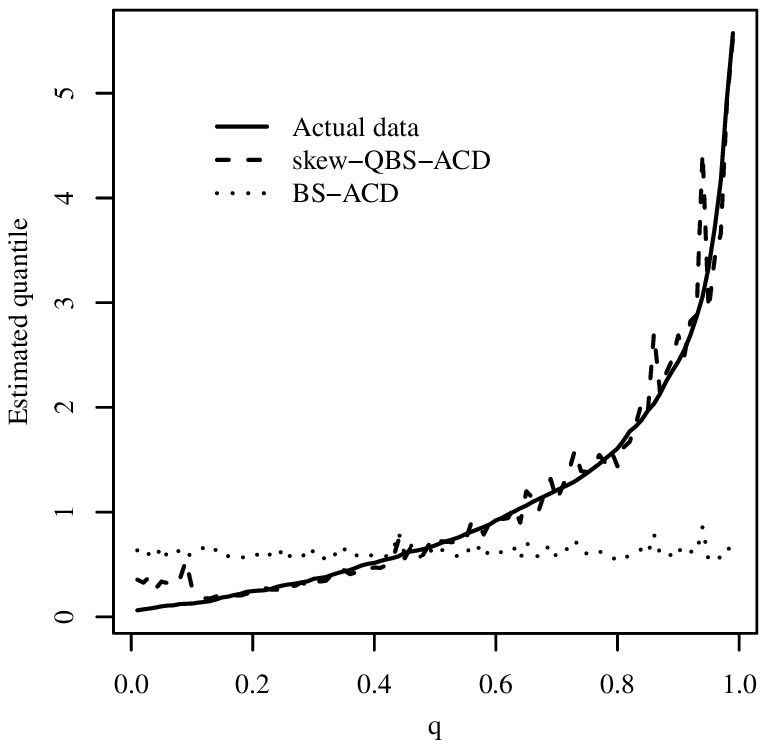}}
\subfigure[Out-of-sample]{\includegraphics[height=6.2cm,width=6.2cm]{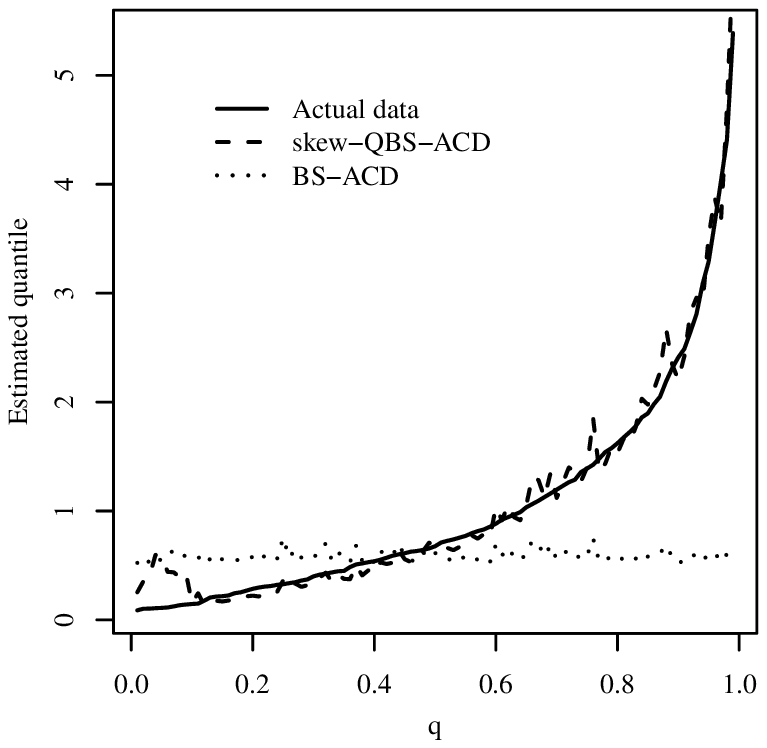}}
\vspace{-0.2cm}
\caption{In-sample and out-of-sample quantile forecasts from the skew-QBS-ACD and BS-ACD models for the BASF-SE data.}\label{figfore:1}
\end{figure}

\section{Concluding remarks}\label{sec:06}
In this paper, we have proposed a quantile autoregressive conditional duration model for dealing with high frequency financial duration data. The proposed model is based on a reparameterized version of the skewed Birnbaum-Saunders distribution, which has the distribution quantile as one of its parameters. Thus, the proposed conditional duration model is defined in terms of a conditional quantile duration, allowing different estimates for different quantiles of interest. A Monte Carlo simulation is carried out to evaluate the performance of the maximum likelihood estimators, study obtained by the ECM algorithm. We have applied the proposed and some other existing models to a real financial duration data set from the German DAX stock exchange. The results are seen to be quite favorable to the proposed model both in terms of model fitting and forecasting ability. The gain in considering the dynamics in terms of conditional quantile duration seems quite evident, especially in the forecasting superiority of the model. As part of future research, it will be of interest to use the theory of estimating functions
for estimating the model parameters \citep{ALLEN2013214}. Furthermore, quantile autoregressive conditional duration models can be used to estimate the intraday volatility of a stock \citep{Tse-2012}. Work on these problems is currently in progress and we hope to report these findings in future papers.


\noindent
\textbf{Funding}: Helton Saulo gratefully acknowledges financial support from FAP-DF.


%

\normalsize


\appendix

\section{Appendix}

\subsection{Some properties of the skew-QBS distribution}

In this subsection, some results on the unimodality property, monotonicity of the hazard function, among others, of the skew-QBS distribution, are obtained.

\begin{prop}[Modes]
A mode of the \textrm{skew-QBS}$(\alpha,\Xi_{q},\lambda)$ is any point $y_0 = y_0(\alpha,\Xi_{q},\lambda)$ that satisfies the following non-linear equation:
{
	\scalefont{0.955}
\begin{multline*}
\lambda \big[a'(y;\alpha,\Xi_{q})\big]^2
\phi\big[\lambda a(y;\alpha,\Xi_{q})\big] 
+
\Phi\big[\lambda a(y;\alpha,\Xi_{q})\big]
\Big\{a''(y;\alpha,\Xi_{q}) - a(y;\alpha,\Xi_{q}) \Big[a'(y;\alpha,\Xi_{q})\big]^2\Big\}=0.
\end{multline*}
}
\end{prop}
\begin{proof}
The proof is trivial and omitted.
\end{proof}

\subsubsection{Unimodality property of the \textrm{skew-QBS} distribution}

The following result provides sufficient conditions to guarantee unimodality of the \textrm{skew-QBS} distribution when the skewness parameter is positive.
\begin{theorem}[Unimodality]
	\label{Unimodality-bimodality}
Let $Y\sim\textrm{skew-QBS}(\alpha,\Xi_{q},\lambda)$ with $\lambda>0$.
If $f'''_{Y}(y;\alpha,\beta,\lambda)\leq 0$ for all $y\geq y_{\rm BS}$, where $y_{\rm BS}$ is the unique mode of the BS PDF satisfying the inequality
\begin{align}\label{hypothesis}
\lambda^2 y_{\rm BS}^3+ \Xi_{q}(2\alpha^2+\lambda^2) y_{\rm BS}^2 + \Xi_{q}^2(3-\lambda^2)y_{\rm BS}-\lambda^2\Xi_{q}^3>0,
\end{align}
then
the \textrm{skew-QBS} PDF is unimodal.
\end{theorem}
\begin{proof}
Let $T\sim \textrm{BS}(\alpha,\Xi_{q})$ be a random variable with BS PDF $f_{T}(y;\alpha,\Xi_{q})$, $y>0$. A simple computation shows that
\begin{align}
f'_{Y}(y;\alpha,\Xi_{q},\lambda)
=&
2\lambda a'(y;\alpha,\Xi_{q})
\phi\big[\lambda a(y;\alpha,\Xi_{q})\big] f_{T}(y;\alpha,\Xi_{q})
\nonumber
\\[0,2cm] 
&+
2 \Phi\big[\lambda a(y;\alpha,\Xi_{q})\big]
f'_{T}(y;\alpha,\Xi_{q});  \label{der1-f}
\\[0,2cm]
f''_{Y}(y;\alpha,\Xi_{q},\lambda) 
=&
2\lambda \phi\big[\lambda a(y;\alpha,\Xi_{q})\big]
\Big[
G(y;\alpha,\Xi_{q},\lambda) f_{T}(y;\alpha,\Xi_{q}) 
+ 
2 a'(y;\alpha,\Xi_{q})  f'_{T}(y;\alpha,\Xi_{q})
\Big] \nonumber
\\[0,2cm]
&+
2 \Phi\big[\lambda a(y;\alpha,\Xi_{q})\big]
f''_{T}(y;\alpha,\Xi_{q}); \label{der2-f}
\end{align}
where $G(y;\alpha,\Xi_{q},\lambda) $ is the function defined as
\begin{align*}
G(y;\alpha,\Xi_{q},\lambda) 
=
a''(y;\alpha,\Xi_{q}) - \lambda^2 a(y;\alpha,\Xi_{q}) \big[a'(y;\alpha,\Xi_{q})\big]^2.
\end{align*}
By using the expression of the first derivative of $a(y;\alpha,\Xi_{q})$ given in \eqref{at-rep}, and the following second and third derivatives of $a(y;\alpha,\Xi_{q})$:
\begin{align*}
\begin{array}{lllll}
\displaystyle
a''(y;\alpha,\Xi_{q})=-\dfrac{1}{4\alpha y^2}\left[\sqrt{t\over \Xi_{q}}+3\sqrt{\Xi_{q}\over y}\right];
&
\displaystyle
a'''(y;\alpha,\Xi_{q})=\dfrac{3}{8\alpha y^3}\left[\sqrt{y\over \Xi_{q}}+5\sqrt{\Xi_{q}\over y}\right];
\end{array}
\end{align*}
we can rewrite $G(y;\alpha,\Xi_{q},\lambda) $ as follows
\begin{align}\label{G-rewritten}
G(y;\alpha,\Xi_{q},\lambda) 
=
-{1\over 8\alpha^3\Xi_{q}^{3/2} y}\,
\Big[
\lambda^2 y^3+ \Xi_{q}(2\alpha^2+\lambda^2) y^2 + \Xi_{q}^2(3-\lambda^2)y-\lambda^2\Xi_{q}^3
\Big].
\end{align}

In what follows we highlight the following remarks:
\begin{itemize}
	\item[(a)] 
For $\lambda>0$, $a'(y;\alpha,\Xi_{q})>0$ and 
$f'_{T}(y_{\rm BS};\alpha,\Xi_{q})=0$, by \eqref{der1-f}, we have
\begin{align*}
f'_{Y}(y_{\rm BS};\alpha,\Xi_{q},\lambda)
=&
2\lambda a'(y_{\rm BS};\alpha,\Xi_{q})
\phi\big[\lambda a(y_{\rm BS};\alpha,\Xi_{q})\big] f_{T}(y_{\rm BS};\alpha,\Xi_{q})
>0.
\end{align*}
\item[(b)] 
Since the unique mode of the BS PDF, $y_{\rm BS}$,  satisfies the inequality \eqref{hypothesis}, from \eqref{G-rewritten} it follows that $G(y_{\rm BS};\alpha,\Xi_{q},\lambda)<0$. Hence, since $f'_{T}(y_{\rm BS};\alpha,\Xi_{q})=0$ and 
$f''_{T}(y_{\rm BS};\alpha,\Xi_{q})<0$, by \eqref{der2-f}, we get
\begin{align*}
f''_{Y}(y_{\rm BS};\alpha,\Xi_{q},\lambda) 
=&
2\lambda \phi\big[\lambda a(y_{\rm BS};\alpha,\Xi_{q})\big]
G(y_{\rm BS};\alpha,\Xi_{q},\lambda) f_{T}(y_{\rm BS};\alpha,\Xi_{q}) 
\\[0,2cm]
&+
2 \Phi\big[\lambda a(y_{\rm BS};\alpha,\Xi_{q})\big]
f''_{T}(y_{\rm BS};\alpha,\Xi_{q})<0.
\end{align*}
\end{itemize}
Note that statements in Items (a) and (b), in addition to the hypothesis
$f'''_{Y}(y;\alpha,\Xi_{q},\lambda)\leq 0$ for all $y\geq y_{\rm BS}$, guarantee that the function $f'_{Y}(y;\alpha,\Xi_{q},\lambda)$ is decreasing and concave down.
Hence, it follows that the function $f'_{Y}(y;\alpha,\Xi_{q},\lambda)$ has a unique root, denoted by $y_0$, in the interval $(y_{\rm BS}, +\infty)$. On the other hand, by unimodality of the BS distribution \citep{balakundu:19}, note that $f_{T}(y;\alpha,\Xi_{q})$ is increasing on $(0,y_{\rm BS})$, implying that the \textrm{skew-QBS} PDF  $f_{Y}(y;\alpha,\Xi_{q},\lambda)$ is also increasing on the interval $(0,y_{\rm BS})$. Therefore, we obtain that $y_0$ is the unique critical point for the function  $f_{Y}(y;\alpha,\Xi_{q},\lambda)$, $y>0$.
Finally, since $\lim_{y\to 0^+}f_{Y}(y;\alpha,\Xi_{q},\lambda)=0$ and $\lim_{y\to +\infty}f_{Y}(y;\alpha,\Xi_{q},\lambda)=0$, we conclude that $y_0$ is the unique mode of the \textrm{skew-QBS} distribution.
\end{proof}

\subsubsection{Monotonicity of HR of the \textrm{skew-QBS} distribution}
The survival and hazard functions, denoted by SF and HR of the \textrm{skew-QBS} distribution
are given by
$
S_Y(y;\alpha,\Xi_{q},\lambda) = 1-\int_{0}^{y}f_{Y}(\xi;\alpha,\Xi_{q},\lambda)\, {\rm d}\xi$, $y>0$,
and
\begin{align*}
H_Y(y;\alpha,\Xi_{q},\lambda)&={f_{Y}(y;\alpha,\Xi_{q},\lambda)\over S_Y(y;\alpha,\Xi_{q},\lambda)}, \quad y>0,
\end{align*}
respectively.

\begin{prop}[Monotonicity of HR]
\label{Monotonicity-1}
Under the conditions of Theorem \ref{Unimodality-bimodality},  the HR of the \textrm{skew-QBS} distribution has the following monotonic properties:
	\begin{itemize}
		\item[1)] It is increasing for all $y<y_0$;
		\item[2)] It is decreasing for all $y>y_*$, for some $y_*\geq y_0$;
	\end{itemize}
where $y_0$ is the mode of the \textrm{skew-QBS} density.
\end{prop}
\begin{proof}
By unimodality of the \textrm{skew-QBS} density (Theorem \ref{Unimodality-bimodality}), the \textrm{skew-QBS} density $f_Y$ increases on $(0,y_0)$. Then, since the survival function $S_Y(y;\alpha,\Xi_{q},\lambda)$, $y>0$, is a decreasing function it follows that $H_Y(y;\alpha,\Xi_{q},\lambda)$  is a product of two nonnegative increasing functions.
Therefore, it is a increasing function for all $y<y_0$, and the proof of first item follows.	
	
In order to prove the second item, by unimodality (Theorem \ref{Unimodality-bimodality}), we have the existence of a point $y_*\geq y_0$ such that the \textrm{skew-QBS} density is a concave upward function on the interval $J=(y_*,+\infty)$. That is, $f_Y''(y;\alpha,\Xi_{q},\lambda)>0$ for all $Y\in J$. Hence,  $f_Y'(y;\alpha,\Xi_{q},\lambda)$ is a decreasing function on $J$. Furthermore, by unimodality, the \textrm{skew-QBS}  density $f_Y$ decreases on this interval. Then the function $G_Y(y;\alpha,\Xi_{q},\lambda)$, defined as
	\begin{align*}
	W_Y(y;\alpha,\Xi_{q},\lambda)=-\, {f_Y'(y;\alpha,\Xi_{q},\lambda)\over f_Y(y;\alpha,\Xi_{q},\lambda)},
	\end{align*}
	decreases for all $y\in J$, because $W_Y(y;\alpha,\Xi_{q},\lambda)$, $y\in J$, is a product of nonnegative
	decreasing functions. Hence, \cite{gla:80} guarantees  that the HR $H_Y(y;\alpha,\Xi_{q},\lambda)$ is a decreasing function for all $y>y_*$. This completes the proof of Item 2).
	We thus complete the proof of proposition.	
\end{proof}

\subsubsection{Other properties of the \textrm{skew-QBS} distribution}

Let $\Theta_q=\sqrt{2}{\rm erf}^{-1}(q)$ be the $100q$-th quantile of the 
HN distribution, where ${\rm erf}^{-1}$ is the inverse error function.
The following proposition shows that the $\textrm{skew-QBS}(\alpha,\Xi_{q},\lambda)$ 
distribution approaches the truncated $\textrm{BS}(\alpha,\Theta_{q})$ distribution with support $[\Theta_q,\infty)$, when the 
skewness parameter $\lambda$ tends to infinity.
\begin{prop}
	If $Y\sim\textrm{skew-QBS}(\alpha,\Xi_{q},\lambda)$,
	the following properties hold:
	\begin{enumerate}
		\item 
		$
		\lim\limits_{\lambda\to\infty} f_{Y}(y;\alpha,\Xi_{q},\lambda)
		=
		2\phi\big[a(y;\alpha,\Theta_q)\big] a'(y;\alpha,\Theta_q)\, I_{[\Theta_q,\infty)}(y), \
		$
		for all $y\neq \Theta_q$;
		\item 
		$
		\lim\limits_{\lambda\to\infty} f_{Y}(\Theta_q;\alpha,\Xi_{q},\lambda)
		=
		\big(\alpha \Theta_q\sqrt{2\pi}\,\big)^{-1}.
		$
	\end{enumerate}
\end{prop} 
\begin{proof}
The proof is immediate upon observing that $\lim_{\lambda\to\infty} \Xi_{q}=\Theta_q$, 
$\lim_{\lambda\to\infty} a(y;\alpha,\Xi_{q})=a(y;\alpha,\Theta_q)$ and
$\lim_{\lambda\to\infty} a'(y;\alpha,\Xi_{q})=a'(y;\alpha,\Theta_q)$.
\end{proof}
\begin{prop}\label{prop:proper}
Let $Y\sim\textrm{skew-QBS}(\alpha,\Xi_{q},\lambda)$. Then:
\begin{enumerate}
\item $cT\sim\textrm{skew-QBS}(\alpha,c\,\Xi_{q},\lambda)$, with $c>0$;
\item $T^{-1}\sim   \textrm{skew-QBS}(\alpha,\Xi_{q}^{-1},-\lambda)$.
\end{enumerate}
\end{prop}
\begin{proof} The proof follows immediately through transformation of variables.	
\end{proof}

\section{Appendix}

The elements of the Hessian matrix are given by
\begin{align*}
{\partial^2{\ell}(\bm{\theta})\over \partial\gamma  \partial\vartheta}
=&
-
\sum_{t=1}^{n}
\dfrac{1}{f^2_{Y|\Xi_{q}}(y_{t};\alpha,\Xi_{q,t},\lambda)} \,
\dfrac{\partial f_{Y|\Xi_{q}}(y_{t};\alpha,\Xi_{q,t},\lambda)}{\partial \gamma} \,
\dfrac{\partial f_{Y|\Xi_{q}}(y_{t};\alpha,\Xi_{q,t},\lambda)}{\partial \vartheta}
\\[0,2cm]
&+
\sum_{t=1}^{n}
\dfrac{1}{f_{Y|\Xi_{q}}(y_{t};\alpha,\Xi_{q,t},\lambda)} \,
\dfrac{\partial^2 f_{Y|\Xi_{q}}(y_{t};\alpha,\Xi_{q,t},\lambda)}{\partial \gamma \partial \vartheta},
\end{align*}
for each $\gamma,\vartheta\in\{\alpha, \varpi, \rho_{1},\ldots,\rho_{r}, \sigma_{1}, \ldots, \sigma_{s}, \lambda\}$,
where
{	\scalefont{0.85}
	\begin{align*}
	&{\partial^2 f_{Y|\Xi_{q}}(y_{t};\alpha,\Xi_{q,t},\lambda)\over \partial\gamma\partial\vartheta^*}
	=
2 \Phi\big[\lambda a(y_t;\alpha,\Xi_{q,t},\lambda)\big] \,	
{\partial^2 f_{T|\Xi_{q}}(y_{t};\alpha,\Xi_{q,t},\lambda)\over \partial\gamma \partial\vartheta^*}
	\\[0,2cm]
&
+2\lambda \phi\big[\lambda a(y_t;\alpha,\Xi_{q,t},\lambda)\big]
\bigg\{	
{\partial a(y_t;\alpha,\Xi_{q,t},\lambda)\over \partial\gamma} \,
{\partial f_{T|\Xi_{q}}(y_{t};\alpha,\Xi_{q,t},\lambda)\over \partial\vartheta^*}
+
{\partial a(y_t;\alpha,\Xi_{q,t},\lambda)\over \partial\vartheta^*} \,
{\partial f_{T|\Xi_{q}}(y_{t};\alpha,\Xi_{q,t},\lambda)\over \partial\gamma}
	\\[0,2cm]
&+
\bigg[	
{\partial^2 a(y_t;\alpha,\Xi_{q,t},\lambda)\over \partial\gamma\partial\vartheta^*}
-
\lambda^2 a(y_t;\alpha,\Xi_{q,t},\lambda) \,
{\partial a(y_t;\alpha,\Xi_{q,t},\lambda)\over \partial\gamma} \,
{\partial a(y_t;\alpha,\Xi_{q,t},\lambda)\over \partial\vartheta^*}
\bigg]
f_{T|\Xi_{q}}(y_{t};\alpha,\Xi_{q,t},\lambda)
	\bigg\}
\end{align*}
}
and 
{	\scalefont{0.85}
	\begin{align*}
	{\partial^2 f_{Y|\Xi_{q}}(y_{t};\alpha,\Xi_{q,t},\lambda)\over \partial\gamma\partial\lambda}
	=
{\partial^2 f_{Y|\Xi_{q}}(y_{t};\alpha,\Xi_{q,t},\lambda)\over \partial\gamma\partial\vartheta^*}	
+ 2 \phi\big[\lambda a(y_t;\alpha,\Xi_{q,t},\lambda)\big] 
f_{T|\Xi_{q}}(y_{t};\alpha,\Xi_{q,t},\lambda)\,
{\partial a(y_t;\alpha,\Xi_{q,t},\lambda)\over \partial\gamma},	
\end{align*}
}
for each $\vartheta^*\in\{\alpha, \varpi, \rho_{1},\ldots,\rho_{r}, \sigma_{1}, \ldots, \sigma_{s}\}$,
with
{	\scalefont{0.85}
\begin{align*}
	&{\partial^2 f_{T|\Xi_{q}}(y_{t};\alpha,\Xi_{q,t},\lambda)\over \partial\gamma\partial\vartheta}
	=
	-a(y_t;\alpha,\Xi_{q,t},\lambda) \,
	{\partial a(y_t;\alpha,\Xi_{q,t},\lambda)\over \partial\vartheta} \, 
{\partial f_{T|\Xi_{q}}(y_{t};\alpha,\Xi_{q,t},\lambda)\over \partial\vartheta}	
	\\[0,2cm]
& 
+	
	\phi\big[a(y_t;\alpha,\Xi_{q,t},\lambda)\big]
\bigg\{
{\partial^2 a'(y_t;\alpha,\Xi_{q,t},\lambda)\over \partial\gamma\partial\vartheta}
-
a(y_t;\alpha,\Xi_{q,t},\lambda)
a'(y_t;\alpha,\Xi_{q,t},\lambda)\,
{\partial^2 a(y_t;\alpha,\Xi_{q,t},\lambda)\over \partial\gamma\partial\vartheta}
	\\[0,2cm]
& 
-
\bigg[
{\partial a(y_t;\alpha,\Xi_{q,t},\lambda)\over \partial\vartheta} \, 
a'(y_t;\alpha,\Xi_{q,t},\lambda)
+
a(y_t;\alpha,\Xi_{q,t},\lambda)\,
{\partial a'(y_t;\alpha,\Xi_{q,t},\lambda)\over \partial\vartheta}
\bigg]
\bigg\},
	\end{align*}
}
for each $\gamma, \vartheta\in\{\alpha, \varpi, \rho_{1},\ldots,\rho_{r}, \sigma_{1}, \ldots, \sigma_{s},\lambda\}$.
A simple computation shows that the above second order partial derivatives of functions $a(y_t;\alpha,\Xi_{q,t},\lambda)$ and $a'(y_t;\alpha,\Xi_{q,t},\lambda)$  are expressed as
\begin{align*}
\textstyle
{\partial^2 a(y_t;\alpha,\Xi_{q,t},\lambda)\over\partial\alpha^2}
&= \textstyle
{a(y_t;\alpha,\Xi_{q,t},\lambda)\over \alpha^2}
-
{1\over \alpha}\, {\partial a(y_t;\alpha,\Xi_{q,t},\lambda)\over \partial\alpha}
+
2 a'(y_t;\alpha,\Xi_{q,t},\lambda) \eta_{\alpha; \lambda} \,
{\partial^2 \eta_{\alpha; \lambda}\over \partial\alpha^2}
\\[0,2cm]
&\quad +
\textstyle
2\big[
{\partial a'(y_t;\alpha,\Xi_{q,t},\lambda)\over \partial\alpha} \,
\eta_{\alpha; \lambda}
+
a'(y_t;\alpha,\Xi_{q,t},\lambda)\, 
{\partial \eta_{\alpha; \lambda}\over \partial\alpha}
\big]
{\partial \eta_{\alpha; \lambda}\over \partial\alpha},
\\[0,2cm]
\textstyle
{\partial^2 a(y_t;\alpha,\Xi_{q,t},\lambda)\over\partial\alpha\partial\gamma^*}
&= \textstyle
- 
{1\over\alpha}\, 
{\partial a(y_t;\alpha,\Xi_{q,t},\lambda)\over \partial\gamma^*}
+
2\, {\partial a'(y_t;\alpha,\Xi_{q,t},\lambda)\over \partial\gamma^*} \, {\partial \eta_{\alpha; \lambda}\over \partial\alpha}\,
\eta_{\alpha; \lambda},
\quad 
\gamma^*\in\{\varpi, \rho_{k},\sigma_{l}\} ,
\end{align*}
\begin{align*}
\textstyle
{\partial^2 a(y_t;\alpha,\Xi_{q,t},\lambda)\over\partial\alpha\partial\lambda}
&= 
\textstyle
-
{1\over \alpha}\, {\partial a(y_t;\alpha,\Xi_{q,t},\lambda)\over \partial\lambda}
+
2 a'(y_t;\alpha,\Xi_{q,t},\lambda) \eta_{\alpha; \lambda} \,
{\partial^2 \eta_{\alpha; \lambda}\over \partial\alpha\partial\lambda}
\\[0,2cm]
&\quad +
\textstyle
2\big[
{\partial a'(y_t;\alpha,\Xi_{q,t},\lambda)\over \partial\lambda} \,
\eta_{\alpha; \lambda}
+
a'(y_t;\alpha,\Xi_{q,t},\lambda)\, 
{\partial \eta_{\alpha; \lambda}\over \partial\lambda}
\big]
{\partial \eta_{\alpha; \lambda}\over \partial\alpha}
,
\\[0,2cm]
\textstyle
{\partial^2 a(y_t;\alpha,\Xi_{q,t},\lambda)\over\partial\gamma^*\partial\vartheta^*}
&= 
\textstyle
\big[
{2\over \Xi_{q,t}}
-
{\partial a'(y_t;\alpha,\Xi_{q,t},\lambda)\over \partial\vartheta^*}
\big]
{1\over \Xi^2_{q,t}}\,
{\partial \Xi_{q,t}\over \partial\gamma^*}
-
{a'(y_t;\alpha,\Xi_{q,t},\lambda)\over \Xi_{q,t}}\, 
{\partial^2 \Xi_{q,t}\over \partial\gamma^*\partial\vartheta^*},
\quad 
\gamma^*,\vartheta^*\in\{\varpi, \rho_{k}, \sigma_{l}\} ,
\\[0,2cm]
\textstyle
{\partial^2 a(y_t;\alpha,\Xi_{q,t},\lambda)\over\partial\gamma^*\partial\lambda}
&= \textstyle
-{1\over \Xi_{q,t}^2}\,
{\partial a'(y_t;\alpha,\Xi_{q,t},\lambda) \over\partial \lambda}\,
{\partial \Xi_{q,t}\over \partial\gamma^*},
\\[0,2cm]
\textstyle
{\partial^2 a(y_t;\alpha,\Xi_{q,t},\lambda)\over\partial\lambda^2}
&= \textstyle
+
2 a'(y_t;\alpha,\Xi_{q,t},\lambda) \eta_{\alpha; \lambda} \,
{\partial^2 \eta_{\alpha; \lambda}\over \partial\lambda^2}
\\[0,2cm]
&\quad +
\textstyle
2\big[
{\partial a'(y_t;\alpha,\Xi_{q,t},\lambda)\over \partial\lambda} \,
\eta_{\alpha; \lambda}
+
a'(y_t;\alpha,\Xi_{q,t},\lambda)\, 
{\partial \eta_{\alpha; \lambda}\over \partial\lambda}
\big]
{\partial \eta_{\alpha; \lambda}\over \partial\lambda},
\end{align*}
where $\eta_{\alpha; \lambda}$ is as in \eqref{definition-eta},
$k=1,\ldots,r$ and $l=1,\ldots,s$,  and
\begin{align*}
\textstyle
{\partial^2 a'(y_t;\alpha,\Xi_{q,t},\lambda)\over\partial\alpha^2}
&= \textstyle
{a'(y_t;\alpha,\Xi_{q,t},\lambda)\over \alpha^2}
-
{1\over \alpha}\, {\partial a'(y_t;\alpha,\Xi_{q,t},\lambda)\over \partial\alpha}
+
2 a''(y_t;\alpha,\Xi_{q,t},\lambda) \eta_{\alpha; \lambda} \,
{\partial^2 \eta_{\alpha; \lambda}\over \partial\alpha^2}
\\[0,2cm]
&\quad +
\textstyle
2\big[
{\partial a''(y_t;\alpha,\Xi_{q,t},\lambda)\over \partial\alpha} \,
\eta_{\alpha; \lambda}
+
a''(y_t;\alpha,\Xi_{q,t},\lambda)\, 
{\partial \eta_{\alpha; \lambda}\over \partial\alpha}
\big]
{\partial \eta_{\alpha; \lambda}\over \partial\alpha},
\\[0,2cm]
\textstyle
{\partial^2 a'(y_t;\alpha,\Xi_{q,t},\lambda)\over\partial\alpha\partial\gamma^*}
&= \textstyle
- 
{1\over\alpha}\, 
{\partial a'(y_t;\alpha,\Xi_{q,t},\lambda)\over \partial\gamma^*}
+
2\, {\partial a''(y_t;\alpha,\Xi_{q,t},\lambda)\over \partial\gamma^*} \, {\partial \eta_{\alpha; \lambda}\over \partial\alpha}\,
\eta_{\alpha; \lambda},
\quad 
\gamma^*\in\{\varpi, \rho_{k},\sigma_{l}\} ,
\\[0,2cm]
\textstyle
\textstyle
{\partial^2 a'(y_t;\alpha,\Xi_{q,t},\lambda)\over\partial\alpha\partial\lambda}
&= 
\textstyle
-
{1\over \alpha}\, {\partial a'(y_t;\alpha,\Xi_{q,t},\lambda)\over \partial\lambda}
+
2 a''(y_t;\alpha,\Xi_{q,t},\lambda) \eta_{\alpha; \lambda} \,
{\partial^2 \eta_{\alpha; \lambda}\over \partial\alpha\partial\lambda}
\\[0,2cm]
&\quad +
\textstyle
2\big[
{\partial a''(y_t;\alpha,\Xi_{q,t},\lambda)\over \partial\lambda} \,
\eta_{\alpha; \lambda}
+
a''(y_t;\alpha,\Xi_{q,t},\lambda)\, 
{\partial \eta_{\alpha; \lambda}\over \partial\lambda}
\big]
{\partial \eta_{\alpha; \lambda}\over \partial\alpha}
,
\\[0,2cm]
\textstyle
{\partial^2 a'(y_t;\alpha,\Xi_{q,t},\lambda)\over\partial\gamma^*\partial\vartheta^*}
&= 
\textstyle
\big[
{2\over \Xi_{q,t}}
-
{\partial a''(y_t;\alpha,\Xi_{q,t},\lambda)\over \partial\vartheta^*}
\big]
{1\over \Xi^2_{q,t}}\,
{\partial \Xi_{q,t}\over \partial\gamma^*}
-
{a''(y_t;\alpha,\Xi_{q,t},\lambda)\over \Xi_{q,t}}\, 
{\partial^2 \Xi_{q,t}\over \partial\gamma^*\partial\vartheta^*},
\quad 
\gamma^*,\vartheta^*\in\{\varpi, \rho_{k}, \sigma_{l}\} ,
\\[0,2cm]
\textstyle
{\partial^2 a'(y_t;\alpha,\Xi_{q,t},\lambda)\over\partial\gamma^*\partial\lambda}
&= \textstyle
-{1\over \Xi_{q,t}^2}\,
{\partial a''(y_t;\alpha,\Xi_{q,t},\lambda) \over\partial \lambda}\,
{\partial \Xi_{q,t}\over \partial\gamma^*},
\\[0,2cm]
\textstyle
{\partial^2 a'(y_t;\alpha,\Xi_{q,t},\lambda)\over\partial\lambda^2}
&= \textstyle
+
2 a''(y_t;\alpha,\Xi_{q,t},\lambda) \eta_{\alpha; \lambda} \,
{\partial^2 \eta_{\alpha; \lambda}\over \partial\lambda^2}
\\[0,2cm]
&\quad +
\textstyle
2\big[
{\partial a''(y_t;\alpha,\Xi_{q,t},\lambda)\over \partial\lambda} \,
\eta_{\alpha; \lambda}
+
a''(y_t;\alpha,\Xi_{q,t},\lambda)\, 
{\partial \eta_{\alpha; \lambda}\over \partial\lambda}
\big]
{\partial \eta_{\alpha; \lambda}\over \partial\lambda},
\end{align*}
for each $k=1,\ldots,r$ and $l=1,\ldots,s$, with
\begin{align*}
\textstyle
{\partial a''(y_t;\alpha,\Xi_{q,t},\lambda)\over\partial\alpha}
&= \textstyle
-{a''(y_t;\alpha,\Xi_{q,t},\lambda)\over 2\alpha y_t}
+
2 a'''(y_t;\alpha,\Xi_{q,t},\lambda)
\eta_{\alpha; \lambda}\,
{\partial \eta_{\alpha; \lambda}\over \partial\alpha} ,
\\[0,2cm]
\textstyle
{\partial a''(y_t;\alpha,\Xi_{q,t},\lambda)\over\partial\gamma^*}
&= \textstyle
-
a'''(y_t;\alpha,\Xi_{q,t},\lambda) \,
{1\over \Xi_{q,t}^2} \,
{\partial \Xi_{q,t}\over \partial \gamma^*},
\quad 
\gamma^*\in\{\varpi, \rho_{1},\ldots,\rho_{r}, \sigma_{1}, \ldots, \sigma_{s}\} ,
\\[0,2cm]
\textstyle
{\partial a''(y_t;\alpha,\Xi_{q,t},\lambda)\over\partial\lambda}
&= \textstyle
2a'''(y_t;\alpha,\Xi_{q,t},\lambda) \eta_{\alpha; \lambda}
\,
{\partial \eta_{\alpha; \lambda}\over \partial\lambda},
\end{align*}
and 
\begin{align*}
a'''(y_t;\alpha,\Xi_{q,t},\lambda) 
=
{3\over 8\alpha y^3_t}
\left[\sqrt{\frac{y_t \eta_{\alpha; \lambda}^2}{4\Xi_{q,t}}}+5\sqrt{\frac{4\Xi_{q,t}}{y_t\eta_{\alpha; \lambda}^2}}\right].
\end{align*}

\end{document}